\documentclass[11pt,onecolumn]{article}

\usepackage[margin=2cm, top=2.5cm, bottom=2.5cm]{geometry}
\usepackage{amsmath, amssymb, amsthm}
\usepackage{graphicx}
\usepackage{booktabs}
\usepackage[hidelinks]{hyperref}
\usepackage{cleveref}
\usepackage{xcolor}
\usepackage{microtype}
\usepackage{authblk}
\usepackage{natbib}
\usepackage{setspace}
\usepackage{enumitem}
\usepackage{caption}
\usepackage{subcaption}
\usepackage{placeins}
\newcommand{\SSV}{\mathrm{SSV}}

\DeclareMathOperator{\rank}{rank}

\newif\ifshowrevisions
\showrevisionsfalse
\ifshowrevisions
\newcommand{\rev}[1]{#1}
\newcommand{\reviewrev}[1]{\textcolor{red}{#1}}
\else
\newcommand{\rev}[1]{#1}
\newcommand{\reviewrev}[1]{#1}
\fi

\newif\ifshownewrevisions
\shownewrevisionsfalse
\ifshownewrevisions
\newcommand{\newrev}[1]{\textcolor{red}{#1}}
\else
\newcommand{\newrev}[1]{#1}
\fi

\newtheorem{theorem}{Theorem}
\newtheorem{proposition}[theorem]{Proposition}
\newtheorem{corollary}[theorem]{Corollary}
\newtheorem{lemma}[theorem]{Lemma}
\theoremstyle{definition}
\newtheorem{definition}[theorem]{Definition}
\theoremstyle{remark}
\newtheorem{remark}[theorem]{Remark}
\theoremstyle{plain}

\title{%
	\textbf{Reciprocity coordinates a three-way trade-off in recurrent neural circuits}
}

\author[1,2]{Xiaoyu Zhang\textsuperscript{*}}
\author[2]{Mingyi Huang\textsuperscript{*}}
\author[1]{Yixuan Cui}
\author[1]{Xin Lu\textsuperscript{\textdagger}}
\author[2]{Yuguo Yu\textsuperscript{\textdagger}}

\affil[1]{%
	National University of Defense Technology
}
\affil[2]{%
	Fudan University
}

\footnotetext[1]{These authors contributed equally to this work.}
\footnotetext[2]{Corresponding author}

\date{\today}

\begin{document}
	\maketitle
	
\begin{abstract}
		Biological neural circuits contain both unidirectional and reciprocal
		connections. However,
			how the proportion of reciprocal connections affects circuit computation
			remains unclear. \rev{To address this question, we systematically varied
			reciprocity $\rho$ across 21 levels in recurrent-network simulations while
			holding network size, connection density, and total directed-edge count fixed.}
			Low-reciprocity
		networks generated more diverse stimulus-driven activity and performed better
		on four tasks requiring continuous input processing. In contrast, higher
		reciprocity generally improved performance across three delayed-memory tasks
		tested over multiple delay lengths. Two complementary experiments revealed the
		dynamical basis of these opposing task effects.
			\newrev{Measurements on the same recurrent matrices showed that increasing
			reciprocity reduced online memory capacity, a task-specific measure of delayed
			input reconstruction, while accelerating recovery from perturbations under
			the original maximum-real-eigenvalue operating regime.} A
		separate post-cue noise experiment showed that higher reciprocity suppressed
		the propagation of internal state noise while making the states evoked by
			different cues less distinguishable. \rev{Component matching of cue separation
			and within-cue dispersion on the same 210-matrix bank then operationally
			distinguished their contributions to frozen-readout robustness under the
			specified matching protocol.} \rev{Within this tested ensemble, reciprocity
			therefore coordinates a}
		three-way trade-off among representational capacity, local dynamical stability,
		and task-state separability: increasing reciprocity reduces representational
		capacity, strengthens local stability, and compresses task-state separability.
		Whether higher reciprocity improves behaviour depends on whether the gain in
		stability outweighs the loss of separation between task-relevant states.
		Finally, \newrev{connectome and microcircuit datasets} from five species occupied an intermediate-reciprocity
		regime, while reciprocity varied substantially among anatomical regions of the
			adult \textit{Drosophila} brain. \reviewrev{These cross-species and within-brain
			observations provide structural context compatible with the framework and
			motivate circuit-specific functional tests.}
\end{abstract}
	

	\clearpage

	\section{Introduction}
	\label{sec:intro}

	Biological neural circuits contain both unidirectional and reciprocal
	synaptic connections. Reciprocity---the fraction of directed connections whose
	reverse counterpart is also present---is a basic feature of circuit topology.
	Reciprocal connections have been observed across biological systems, from local
	cortical microcircuits to whole-brain connectomes, although reported levels of
	reciprocity vary across systems~\citep{song2005highly,cook2019whole,
	lin2024network,peng2024directed}. Previous studies have linked reciprocal
	structure to computational properties including activity dimensionality,
	dynamical timescales, and memory performance. How these properties interact as
	reciprocity changes, however, remains without a unified account.

	The encoding of input depends on representational capacity: the ability of a
	network to generate diverse activity states and use them to encode information.
	Multi-patch recordings from human cortical Layer~2/3 revealed a directed and
	largely acyclic excitatory microcircuit in which reciprocity was close to that
	expected under random connectivity. Models incorporating these wiring
	principles exhibited higher-dimensional neural dynamics and improved
	performance on selected tasks~\citep{peng2024directed}. More generally,
	feedforward and functionally feedforward recurrent networks can preserve recent
	inputs through sequential transient activity without requiring a stable
	attractor~\citep{ganguli2008memory,goldman2009memory}. Motif-based theory further
	identifies reciprocal motifs as important local determinants of activity
	dimensionality in recurrent networks~\citep{recanatesi2019dimensionality}. The
	direction of this effect is not fixed: reciprocal motifs reduce covariance
	dimensionality when absolute coupling strength is held constant, but increase
	it when coupling is controlled relative to the stability
	boundary~\citep{hu2022spectrum}. Recent density-conserving manipulations of
	reciprocity likewise found that lower reciprocity increased memory capacity and
	representational diversity in the reservoir networks examined
	there~\citep{hadaeghi2026controlling,hadaeghi2025computational}. Together, these
	results show that reciprocal structure reshapes the state space explored by
	network activity, although the direction of the effect depends on the model and
	control conditions.
	\newrev{Extending this motif--dimensionality account to cortical data,
	\citet{dahmen2026localized} combined neural population recordings with synaptic
	physiology and theory to relate strong recurrence and local motifs to activity
	dimensionality. Their analysis also shows that motif effects depend on the
	excitatory and inhibitory populations engaged, linking anatomical structure
	to state-dependent effective connectivity.}

	Reciprocal structure affects not only which activity states a network can
	generate, but also how those states evolve over time. A related property is
	local dynamical stability: the tendency of a network state to recover from a
	small perturbation and to amplify or attenuate internally generated
	fluctuations. Symmetric recurrent connections provide the classical basis for
	attractor dynamics, allowing activity patterns to persist after external input
	has disappeared and to return after local
	perturbations~\citep{hopfield1982neural}. This mechanism provides a classical
	account of the persistent delay-period activity observed during working-memory
	behaviour~\citep{fuster1971neuron,goldman1995cellular}. In random-network
	models, correlations between reciprocal synaptic weights can slow intrinsic
	dynamics~\citep{marti2018correlations}, while feedback generated by connectivity
	motifs can extend the timescale over which a network retains past
	inputs~\citep{hu2018feedback}. Positive reciprocal motifs in
	excitatory--inhibitory networks can also strengthen recurrent feedback and
	induce bistable population states~\citep{shao2023relating,shao2025impact}.
	Broader studies of motifs and partially symmetric networks similarly link
	reciprocal structure to network feedback, dynamical timescales, and fixed-point
	organisation~\citep{hu2013motif,hu2014local,sherf2025complexity}. Overall,
	reciprocal structure can alter the maintenance and recovery of network states
	by regulating recurrent feedback, dynamical timescales, and attractor
	organisation.
	\newrev{At the inter-area level, \citet{javadzadeh2026consensus} combined
	visual-cortex recordings, optogenetic perturbations, and circuit modeling to
	show that reciprocal excitatory feedback selectively preserves congruent
	activity while suppressing disagreement between V1 and LM. This illustrates
	how feedback can support persistent modes and consensus through selective,
	rather than uniform, stabilization.}

	Whether maintained states support reliable task information also depends on
	whether states induced by different task cues remain distinguishable. We refer
	to the degree of distinction among cue-conditioned network states as
	task-state separability. In a direct comparison of feedforward and attractor
	memory networks, Fisher-information analysis showed that memory performance
	depends jointly on the distance between the mean states associated with
	different stimuli and the noise variance within each state: reducing the former
	or increasing the latter diminishes the stimulus information retained by the
	network~\citep{lim2012noise}. Wainrib likewise showed that connectivity symmetry
	changes how much past context is retained in the current network state, with the
	direction of the effect depending on the nonlinear operating
	regime~\citep{wainrib2015context}. Recent reservoir studies further showed that
	reciprocal motifs can redistribute memory between recent and more distant
	inputs~\citep{takasu2025neuronal}. Overall, the retention of task information
	depends on both the separation among input-conditioned states and their
	within-state variability, and connectivity symmetry and reciprocal motifs can
	alter this information-retention process.
	\newrev{Complementarily, \citet{ghosh2026partial} found that the effects of
	individual pathways in task-trained partially recurrent networks depend on
	both circuit architecture and task, with feedback improving noise robustness
	only in some contexts. Pathway-level partial recurrence differs from the
	neuron-level reciprocal-edge fraction studied here, but reinforces the need
	to evaluate feedback across distinct computational demands.}

	\newrev{These studies establish that feedback can shape activity geometry,
	stability, and computation, with effects that depend on circuit structure and
	operating regime. The question addressed here is how varying reciprocal-edge
	allocation under a fixed wiring budget jointly changes activity geometry,
	stability, and cue separation, and how these changes relate to distinct task
	demands.} We propose that
	reciprocity coordinates a three-way trade-off among representational capacity,
	local dynamical stability, and task-state separability.
	\newrev{We distinguish activity geometry, measured by SSV, from performance on
	driven tasks. In the shared-matrix analysis, the coordinate $C$ denotes summed
	online memory capacity; it is one specific temporal-reconstruction measure,
	not a definition of representational capacity.}

		\rev{To test this hypothesis, we systematically varied reciprocity across 21
		levels in echo state networks} while holding network size, connection density, and total directed-edge
	count fixed, and evaluated four continuously driven tasks and three
	delayed-memory tasks. The functional effect of reciprocity was strongly task
	dependent: lower reciprocity supported richer stimulus-driven representations
	and better performance on continuously driven tasks, whereas increasing
	reciprocity generally benefited delayed memory. Further analyses related these
	differences to the proposed three-way trade-off among representational capacity,
	local dynamical stability, and task-state separability. \newrev{Connectome and
	microcircuit datasets} from five
	species occupied an intermediate-reciprocity regime, while reciprocity varied
		substantially across anatomical regions of the adult \textit{Drosophila} brain,
		\reviewrev{providing structural context for the model framework.} Overall, reciprocity does not
	determine a universally optimal circuit architecture; instead, it regulates how
	circuit dynamics are distributed among these three computational properties.
	
	\section{Methods}
	\label{sec:methods}
	
	\subsection{Controlled generation of reciprocal recurrent networks}
	\label{sec:methods:network}
	
	All reservoir networks consist of $N = 300$ neurons unless otherwise stated.
	The recurrent weight matrix $\mathbf{W} \in \mathbb{R}^{N \times N}$ is
	constructed in three stages, enabling independent control of reciprocity,
	excitation--inhibition (E/I) balance, and \rev{spectral normalization}.
	
	\paragraph{Reciprocity-controlled topology.}
		\rev{We define the target reciprocity parameter $\rho \in [0, 1]$ as the
		intended fraction of directed edges that possess a reciprocal counterpart.}
		The adjacency structure is
	generated by the generate\_matrix\_reciprocity algorithm:
		\rev{from $N(N-1)$ potential directed edges, the algorithm sets
		$E = \operatorname{round}[N(N-1)c]$, where $c$ is the target connection
		probability ($c = 0.1$ throughout). It then sets
		$M=\lfloor E\rho/2\rfloor$ reciprocal unordered dyads. These dyads consume
		$2M$ directed edges, placed as pairs $(i \rightarrow j,j \rightarrow i)$
		with independently drawn positive weights
		$w_{ij}\sim\mathcal{U}(0,1)$. The remaining $E-2M$ edges are placed
		unidirectionally, with direction randomised by a coin flip.}
		At $\rho = 0$, the network is a fixed-edge-count random oriented graph:
		each selected unordered pair contains exactly one randomly oriented edge, so
		accidental reciprocal pairs are excluded. It is therefore not a standard
		directed Erd\H{o}s--R\'{e}nyi graph. At $\rho = 1$, every selected edge belongs
		to a reciprocal pair, although the two weights remain independent.
		Because the construction holds the number of directed excitatory edges fixed,
		reciprocity is algebraically coupled to distinct-partner coverage. Let $M$ be
		the number of reciprocal dyads and let the realized reciprocity be
		$\rho_{\mathrm r}=2M/E$. The number of occupied unordered neuron pairs is
		\begin{equation}
			P=E-M=E\left(1-\frac{\rho_{\mathrm r}}{2}\right).
			\label{eq:partner_budget}
		\end{equation}
		Each occupied pair contributes two neuron--partner incidences. Defining the
		realized directed density $c_{\mathrm r}=E/[N(N-1)]$ and normalized
		distinct-partner coverage $U=2P/[N(N-1)]$ therefore gives the exact identity
		\begin{equation}
			U=c_{\mathrm r}(2-\rho_{\mathrm r}).
			\label{eq:partner_coverage}
		\end{equation}
		Thus a reciprocal dyad spends two directed edges on one partner pair, whereas
		two unidirectional edges can reach two distinct partner pairs. These identities
		describe the finite-budget ensemble used in this study; they do not imply that
		reciprocity and partner coverage must be coupled in fixed-skeleton or
		degree-preserving ensembles.
	
	\paragraph{Excitation--inhibition balance.}
		The reciprocity-controlled matrix $\mathbf{W}_e$ (all weights positive) is
		balanced by an independently sampled inhibitory component $\mathbf{W}_i$.
		Its connection probability is computed from an assumed excitatory/inhibitory
		ratio of $60/40$ and is capped at one; with the parameters used here this cap
		is active, making $\mathbf{W}_i$ dense off the diagonal. The final signed weight matrix is
	$\mathbf{W} = \alpha_e \mathbf{W}_e - \alpha_i \mathbf{W}_i$, with scaling
	factors $\alpha_e = s_i / (s_e + s_i)$ and $\alpha_i = s_e / (s_e + s_i)$
	($s_e, s_i$ the summed excitatory and inhibitory weights) chosen
	algorithmically so that the net summed weight cancels,
	$\sum_{ij} W_{ij} \approx 0$. This construction is motivated by the broader
	balanced-network principle that large excitatory and inhibitory contributions
	can cancel~\citep{vanvreeswijk1996chaos}; it is not the dynamically emergent,
	neuron-level balanced state studied in that work. We therefore refer to it as
	a global E/I-cancellation signed-weight scheme. It neither enforces row-wise
	balance nor strict neuron-level Dale's law.
	
	\paragraph{Spectral normalization and operating-regime scope.}
	To control the nominal dynamical regime, $\mathbf{W}$ is scaled to a target
	gain $\sigma = 1.0$. The main simulations use the code's norm mode
		\texttt{'r'}, which computes the spectral abscissa
		$a(\mathbf{W})=\max_i\operatorname{Re}\lambda_i(\mathbf{W})$ and scales
		$\mathbf{W} \leftarrow \mathbf{W}\cdot \sigma/a(\mathbf{W})$. The generated
		matrices satisfy $a(\mathbf{W})>0$; otherwise this normalization would be
		undefined. We call this convention \emph{maximum-real} normalization, not
		spectral-radius normalization. It was the original, hypothesis-generating
		code path and was fixed before the paired normalization control was examined.
		Its rationale is to hold the largest real-axis feedback component constant
		while reciprocity changes, but it is not the conventional ESN stability
		constraint~\citep{jaeger2001echo}. In particular, the maximum eigenvalue
	modulus $|\lambda|_{\max}$ can exceed unity, so this convention does not by
	itself establish the echo-state property of the zero-state linearization.
	We therefore monitor $|\lambda|_{\max}$ and repeat the shared-matrix \newrev{online memory capacity}
	and restoration analyses after scaling the same pre-normalized matrices to
	maximum modulus $0.95$. Claims that fail this paired, conventional control are
	explicitly restricted to the original maximum-real operating regime.
	
	\subsection{Reservoir dynamics}
	\label{sec:methods:dynamics}
	
	The reservoir is a discrete-time leaky-integrator echo state network. The
	state vector $\mathbf{x}(t) \in \mathbb{R}^{N}$ evolves according to
	\begin{equation}\label{eq:esn}
		\mathbf{x}(t+1) = (1 - \alpha)\,\mathbf{x}(t)
		+ \alpha\,\tanh\bigl(\mathbf{W}\,\mathbf{x}(t)
		+ \mathbf{W}_{\mathrm{in}}\,\tilde{\mathbf{u}}(t)\bigr),
	\end{equation}
	where $\alpha \in (0, 1]$ is the leak rate (inverse time constant) and
	$\tilde{\mathbf{u}}(t) = [1; \mathbf{u}(t)]$ augments the input vector
	$\mathbf{u}(t) \in \mathbb{R}^{d_{\mathrm{in}}}$ with a constant bias term.
	The input weight matrix $\mathbf{W}_{\mathrm{in}} \in \mathbb{R}^{N \times
		(d_{\mathrm{in}}+1)}$ is initialized with random entries drawn uniformly from
	$(-0.5, 0.5)$ and remains fixed throughout. The leak rate is task-dependent:
	$\alpha = 0.1$ for SAD, $\alpha = 0.2$ for MC and MSO, and $\alpha = 0.3$ for MG, chosen to match
	the characteristic timescales of each dataset.
	
	\paragraph{Readout training.}
	For time-series forecasting (MG), the readout weights
	$\mathbf{w}_{\mathrm{out}}$ are trained via ridge regression with
	regularisation $\lambda = 10^{-5}$ on the reservoir states collected over the
	training sequence. For classification tasks, a linear discriminant
	analysis (LDA) classifier is fitted to the final reservoir state
	$\mathbf{x}(T)$ of each trial, where $T$ is the input duration.
	
	\subsection{Computational tasks and datasets}
	\label{sec:methods:tasks}
	
	We evaluate each reservoir on four driven transient-processing tasks spanning spoken digit
classification, chaotic time-series forecasting, temporal memory capacity, and
	multi-frequency oscillator prediction.
	
	\subsubsection{Spoken Arabic Digit classification (SAD)}
	\label{sec:methods:sad}
	
	We use the Spoken Arabic Digits dataset\newrev{~\citep{bedda2008spoken}}, which contains
	mel-frequency cepstral coefficient (MFCC) features extracted from utterances
	of ten digits (0--9) spoken by male and female speakers. The training set
	\newrev{comprises 6,600 utterances and the test set 2,200 utterances. The original
	dataset contains 88 speakers (44 male and 44 female), each repeating the ten
	digits ten times; the per-sex, per-digit sample counts therefore include
	repeated utterances, rather than distinct speakers.} Each utterance is represented as a
	time series of 40 frames, with the first 6 MFCC channels used as input
	($d_{\mathrm{in}} = 6$). The reservoir processes each utterance and the
	final state $\mathbf{x}(40)$ is used as input to an LDA classifier trained to
	discriminate digit identity. Performance is reported as classification
	accuracy (1 $-$ error rate).
	
	\subsubsection{Mackey--Glass chaotic time-series forecasting (MG)}
	\label{sec:methods:mg}
	
	The Mackey--Glass system~\citep{mackey1977oscillation} is a classic benchmark for
	chaotic time-series prediction, governed by the delay-differential equation
	$\dot{y}(t) = \beta y(t-\tau) / (1 + y(t-\tau)^n) - \gamma y(t)$ with
	parameters $\tau = 17$, $\beta = 0.2$, $n = 10$, $\gamma = 0.1$. We use the
	standard discretised sequence (unit sampling interval, 10,000 time steps). A
	training window of 2,000 steps is followed by a 100-step washout. The
	reservoir is trained to predict $y(t+1)$; performance is evaluated via
	autonomous 40-step iterative forecasting across 100 test windows, with a
	step interval of 5. The primary metric is normalized mean squared error
	(NMSE), computed as $\mathrm{NMSE} = \langle e^2 \rangle / \langle y^2
	\rangle$ where $e$ is the 40-step-ahead prediction error and $\langle \cdot
	\rangle$ denotes the temporal average over the test set.
	
	\subsubsection{Memory Capacity (MC)}
	\label{sec:methods:mc}
	
	The memory capacity task\newrev{~\citep{jaeger2002short}} measures the reservoir's ability
	to retain past input information across time delays. The network is driven by
	a white-noise input $u(t) \sim \mathcal{U}(-0.8, 0.8)$, and for each delay
	$k \in \{1, \dots, K\}$ (where $K = 20$), a linear readout is trained via
	ridge regression to reconstruct $u(t-k)$ from the reservoir state
	$\mathbf{x}(t)$. The memory capacity for delay $k$ is
	$\mathrm{MC}_k = \max(0, 1 - \langle e_k^2 \rangle / \langle u^2 \rangle)$,
	where $e_k$ is the reconstruction error. \newrev{The summed memory score is
	$\mathrm{MC} = \sum_{k=1}^{K} \mathrm{MC}_k$, a task-specific measure of how well
	past scalar inputs can be linearly reconstructed over the tested delays.
	It is distinct from SSV, which characterizes activity-response geometry:
	reconstruction also depends on how the input history is aligned with that
	geometry, and the SSV theory does not establish a universal ordering of MC.}
	Unlike the delayed-maintenance tasks below, MC is measured while the reservoir
	continues to receive a fluctuating input. It therefore quantifies online trace
	memory encoded in the driven state rather than autonomous maintenance of a
	cue-specific state through a silent delay~\citep{ganguli2008memory}.
	
	\subsubsection{Multi-Sine Oscillator prediction (MSO)}
	\label{sec:methods:mso}
	
	The MSO task tests the network's ability to process multiple incommensurate
	temporal frequencies simultaneously. We generate a synthetic input signal as
	a superposition of five sinusoids with incommensurate frequencies
	$f \in \{1/2.1, 1/3.2, 1/5.7, 1/11.3, 1/17.9\}$:
	$u(t) = \frac{1}{5}\sum_{j=1}^{5} \sin(2\pi f_j t)$. The network is trained
	via ridge regression to perform one-step-ahead prediction
	$\hat{y}(t) = u(t+1)$. Performance is evaluated on a held-out test set and
	reported as NMSE. Successful prediction requires the reservoir to maintain
	multiple independent temporal representations simultaneously, making this
	task highly sensitive to the dimensionality collapse induced by reciprocal
	connectivity.
	
	\subsubsection{Adding Problem}
\label{sec:methods:adding}

The Adding Problem~\citep{hochreiter1997long} is a classic benchmark for
long-term temporal dependencies. A sequence of length $T$ is presented; two
random positions are marked with indicator flags, and the corresponding values
at those positions must be remembered and summed at the end of the sequence.
All other positions contain zero-mean noise. We test sequence lengths
$T \in \{100, 200, 400\}$. Performance is reported as normalised mean squared
error (NMSE; lower is better). This task demands that the network maintain
precise scalar information through extended delays while ignoring distractors.

\subsubsection{Delayed Match-to-Sample (DMS)}
\label{sec:methods:dms}

DMS is the canonical working-memory paradigm in systems
neuroscience~\citep{fuster1971neuron,goldman1995cellular}. A 4-class stimulus is
briefly presented (5 time steps), followed by a long delay period (100--400
steps) during which weak additive Gaussian noise ($\sigma = 0.05$) is injected.
The network must then report the identity of the remembered stimulus. Stimulus
amplitude is set to 0.3 to avoid ceiling effects. Classification is performed
by a linear discriminant analysis (LDA) readout trained on 400 trials and
evaluated on 200 held-out trials; accuracy is reported.

\subsubsection{Pattern Hold}
\label{sec:methods:patternhold}

A 10-dimensional binary pattern is briefly presented, followed by a delay
period of $d\in\{50,100\}$ steps without input. The network must reproduce the original
pattern at the end of the delay. Performance is reported as NMSE between the
network output and the target pattern. This task can be viewed as a multidimensional generalisation of the classical
memory capacity (MC) task\newrev{~\citep{jaeger2002short}}: whereas MC measures how long a
single scalar value passively echoes through the network under continuous
driving input, Pattern Hold measures the late linear decodability of a
10-dimensional pattern through a silent delay. Because this assay uses novel
patterns without repeated cue-conditioned trajectories or basin discovery, it
tests silent-delay information retention but does not by itself establish
self-sustained multi-attractor dynamics.

\paragraph{Seeded delayed-memory replication.}
The three delayed-memory tasks were rerun as an independent, fully seeded
replication at 21 reciprocity levels ($\rho=0,0.05,\ldots,1$), with $N=300$,
$c=0.1$, maximum-real-eigenvalue normalization at one, and ten independently
generated recurrent matrices per level (master seed 2026071501). Adding used
500 training and 200 held-out trials per sequence length; DMS used 400 and 200;
Pattern Hold used 300 and 150. Adding and Pattern Hold readouts were ridge
regressions ($\lambda=10^{-4}$). For each network-index block, task stimuli,
targets, and input weights used the same derived seeds across $\rho$, whereas
the recurrent matrix had a separate seed at every $\rho$; this common-random-
number design isolates recurrent-structure variation without reusing a matrix.
Every matrix, derived seed, SHA-256 checksum, network-level score, held-out
target, and prediction was retained. Figure-level trends are Pearson
correlations of the 21 $\rho$-level means. As an independent-unit sensitivity
analysis, a slope was estimated across $\rho$ within each of the ten seed
blocks and tested against zero by two-sided one-sample $t$ and Wilcoxon signed-
rank tests; Benjamini--Hochberg correction was applied across the seven
reported task--delay conditions.

\paragraph{Post-cue process-noise scan.}
To isolate robustness after information had been written, we reused the 210
audited recurrent matrices from the seeded replication and ran an additional
four-class DMS scan at a 200-step delay. The cue was presented for five steps;
all subsequent external inputs were zero, and independent isotropic Gaussian
state noise was added after each recurrent update during the delay,
$\boldsymbol{\xi}_{t}\sim\mathcal{N}(\mathbf{0},\sigma_x^2\mathbf{I})$, at
$\sigma_x\in\{0,2.5\times10^{-5},5\times10^{-5},10^{-4},
2.5\times10^{-4},5\times10^{-4},10^{-3}\}$. This range was frozen after a
three-network dynamic-range pilot, before inspecting the full $\rho$ scan. For
each network-index block, balanced cue labels, input weights, and standard-
normal noise directions used the same derived seeds across $\rho$; the same
noise directions were rescaled across doses.

	Two complementary readouts were evaluated on 200 held-out trials per dose.
	A ridge classifier ($\lambda=10^{-4}$) was fitted once to 400 noiseless
	training trials and then frozen across the entire dose curve, testing the
	robustness of a fixed decoder. \rev{Separately, a shrinkage-LDA readout was
	fitted to 400 independent trials at each dose, measuring cue information after
	adaptation to the dose-specific noise covariance. At exactly zero process
	noise, repeated trials within each cue class have degenerate covariance; this
	fit was retained only as a diagnostic and omitted from the noise-matched
	curve.} \rev{For accuracy $a(\sigma_x)$ and four-class chance
	$a_{\mathrm{ch}}=0.25$, frozen-readout retention was defined as
	$[a(\sigma_x)-a_{\mathrm{ch}}]/[a(0)-a_{\mathrm{ch}}]$ and integrated over the
	full dose range. Noise-matched retention instead used the smallest positive
	dose, $\sigma_{\min}=2.5\times10^{-5}$, as its baseline and was integrated
	over the positive-dose range. Each trapezoidal integral was divided by its
	corresponding dose-range width to obtain a normalized robustness AUC.} We also recorded the dose
at which retention first crossed one half, terminal between- and within-cue
geometry, and the paired empirical noise gain
\begin{equation}
    G_{L}^{\mathrm{emp}}(\sigma_x)
    =\frac{\mathbb{E}\|\mathbf{x}_{L}^{\sigma_x}
      -\mathbf{x}_{L}^{0}\|_2^2}{N\sigma_x^2},
    \label{eq:empirical_noise_gain}
\end{equation}
using identical cue trials with and without process noise. Network-level
Spearman tests were primary; seed-block slopes, paired $\rho$-endpoint contrasts,
and bootstrap confidence intervals were sensitivity analyses. For the
margin/noise-proxy association, we additionally computed a Spearman correlation
within each of the ten common-random-number seed blocks and applied an exact
two-sided sign test to their directions. Matrix hashes, all derived seeds,
predictions, and dose-level network metrics were retained.

	\subsection{Functional metrics of representational capacity}
	\label{sec:methods:ssv}
	
	\paragraph{State-Space Volume (SSV).}
	\newrev{Motivated by covariance-based descriptions of population activity
\citep{cunningham2014dimensionality}, we define the following standardized
Gaussian log-volume index for the reservoir.} Let
	$\mathbf{X}_{c}$ be the centered post-washout activity matrix,
	$\mathbf{G}=\mathbf{X}_{c}\mathbf{X}_{c}^{\top}/T_{\mathrm{eff}}$ its raw
	sample covariance, $\mathbf{D}=\operatorname{diag}(G_{11},\ldots,G_{NN})$,
	and $\mathbf{R}=\mathbf{D}^{-1/2}\mathbf{G}\mathbf{D}^{-1/2}$ its sample
	correlation matrix. The implementation first $z$-scores each neuron's trace,
	so its scatter matrix is $\mathbf{C}=T_{\mathrm{eff}}\mathbf{R}$, and reports
	\begin{equation}\label{eq:ssv}
		\mathrm{SSV} = \frac{N}{2}\log(2\pi) +
		\frac{1}{2}\log\det(\mathbf{C}+\epsilon\mathbf{I}),
	\end{equation}
		where $\epsilon$ denotes numerical regularization. The legacy synthetic
		pipeline uses the absolute numerical eigenvalues without an explicit floor,
		whereas the connectome pipeline replaces eigenvalues below numerical precision
		by $10^{-30}$. For a well-conditioned full-rank spectrum this is equivalent to
		$\epsilon=0$; for a rank-deficient spectrum the result is necessarily
		regularizer-dependent and cannot be treated as an intrinsic volume.
		Equivalently,
	$\mathrm{SSV}=\frac{N}{2}\log(2\pi T_{\mathrm{eff}})+
	\frac{1}{2}\log\det(\mathbf{R}+
	\epsilon\mathbf{I}/T_{\mathrm{eff}})$. Thus the reported scatter-form SSV is a
	scale-invariant generalized variance of the activity correlation matrix, not
	the log-volume of the unstandardized covariance and not an absolute
	differential entropy in physical units. In particular,
	$\log\det\mathbf{R}=\log\det\mathbf{G}-\sum_i\log G_{ii}$: marginal activity
	gains are removed, and SSV measures how independently the neurons' temporal
	response signatures span state space. The exact linear-system interpretation
	and its assumptions are given in Supplementary Note~S5.

	For the synthetic reservoir
	sweep and reviewer controls we use $T = 50{,}000$ time steps with a
	$1{,}000$-step washout under a standard-normal scalar input $u(t)$, whereas
		large connectome runs use the same estimator at a shorter fixed sequence length
		($T = 3{,}000$, $500$-step washout) for tractability. Because
	$T_{\mathrm{eff}}$ is fixed within each comparison, changes in SSV reflect
	changes in the activity correlation spectrum; lower SSV indicates that activity
	has collapsed onto a lower-volume manifold. The zero crossing of SSV has
	no intrinsic physical meaning; because SSV is a log-volume defined only up
	to an additive constant that depends on $T_{\mathrm{eff}}$, only relative
		comparisons at fixed $T_{\mathrm{eff}}$ are interpretable. In addition, centering
		implies $\rank(\mathbf{R})\leq T_{\mathrm{eff}}-1$; therefore an
		unregularized determinant is identically zero whenever
		$N\geq T_{\mathrm{eff}}$. Such connectomes require either a shrinkage estimator
		or a matched-dimensional subgraph analysis and are excluded from quantitative
		SSV interpretation here.
	For descriptive display of the estimable biological connectomes, which differ
	in neuron number $N$, we report the per-neuron value $\mathrm{SSV}/N$ (in nats
	per neuron); we refer to this as normalized SSV. This rescaling does not remove
	$N/T_{\mathrm{eff}}$ bias and is not used for cross-species inference. We
	emphasise that SSV is a log-volume over the full activity
	spectrum and is distinct from a participation-ratio dimensionality (reported
	separately below), which counts only the dominant variance directions.
	
	\paragraph{\newrev{Participation-ratio effective dimension.}}
	As a complementary measure of dimensionality, we compute the participation
	ratio\newrev{~\citep{hu2022spectrum,dahmen2026localized}}:
	\begin{equation}\label{eq:pr}
		\mathrm{PR} = \frac{(\sum_i \lambda_i)^2}{\sum_i \lambda_i^2},
	\end{equation}
	where $\lambda_i$ are the eigenvalues of the activity covariance matrix
	$\boldsymbol{\Sigma}$ after $z$-scoring, yielding the activity-space
	effective rank. For spectral analysis of $\mathbf{W}$, the same formula is
	applied to the absolute eigenvalues $|\lambda_i(\mathbf{W})|$, yielding the
	spectral effective rank, which captures how evenly spectral power is
	distributed across eigenmodes.
	\newrev{This participation-ratio summary differs from the entropy-based
effective rank of \citet{roy2007effective}, which is defined using normalized
singular values.}
	
	\paragraph{Eigenvalue spectrum and spectral analysis.}
	We analyse the eigenvalue spectrum of $\mathbf{W}$ along three dimensions:
	the mean absolute eigenvalue $\langle |\lambda| \rangle$ and its standard
	deviation $\mathrm{std}(|\lambda|)$; the spectral gap $|\lambda_1| -
	|\lambda_2|$ between the two largest eigenvalues, which determines how
	quickly activity collapses onto the dominant eigenmode; and the relative
	commutator normality defect $\delta_{\mathrm{rel}}$ in
	Eq.~\ref{eq:dep}, which vanishes exactly for normal matrices and is invariant
	to overall weight scaling. Additionally, we compute a
	qualitative spectral proxy for the \emph{raw linear covariance} volume from
	the eigenvalues of $\mathbf{W}$:
	$H = -\frac{1}{2}\sum_i \log(1 - |\tilde{\lambda}_i|^2)$, where
	$\tilde{\lambda}_i$ are max-modulus normalized eigenvalues lying inside the
	stable unit disk. This expression is motivated by the covariance of a normal,
	stable linear reservoir under isotropic drive. It is used only as a spectral
	diagnostic: per-neuron standardization introduces an additional diagonal term,
	the exact formula involves the eigenvalues of the leak-adjusted map
	$(1-\alpha)\mathbf{I}+\alpha\mathbf{W}$ rather than those of $\mathbf{W}$
	alone, and the expression is not an exact SSV for either the standardized linear
	system or the nonlinear leaky-tanh dynamics.
	
	\subsection{Topological characterization}
	\label{sec:methods:topology}
	
	For each generated synthetic network, graph-theoretic metrics are computed
	from the pre-balance excitatory support
	$\mathbf{A}=(\mathbf{W}_{e}>0)$. This distinction is necessary because the
	independent inhibitory background is dense under the present parameters;
	spectral and dynamical metrics, by contrast, are computed from the final signed
	and normalized matrix $\mathbf{W}$.
	The measured reciprocity $\hat{\rho} = (\sum_{ij} A_{ij} A_{ji}) / (\sum_{ij}
	A_{ij})$ quantifies the realized fraction of reciprocated edges and is used
	to verify the target $\rho$. Connection density is $c = (\sum_{ij} A_{ij}) /
	(N(N-1))$. \newrev{For community detection we use the undirected binary
projection $B_{ij}=\max(A_{ij},A_{ji})$. Modularity $Q$ follows the undirected
definition \citep{newman2004finding} and is maximised by greedy community
detection \citep{clauset2004finding}: $Q = \frac{1}{2m}\sum_{ij}[B_{ij} -
\frac{k_i k_j}{2m}]\delta(c_i, c_j)$, where $m$ and $k_i$ are the edge count and
node degree in this projection.} Here $\delta(c_i, c_j) = 1$ if nodes $i$ and
	$j$ belong to the same module. Within- and between-module reciprocity
	($\rho_{\mathrm{in}}$ and $\rho_{\mathrm{out}}$) are computed by restricting
	the reciprocity summation to edge pairs whose source and target belong to the
	same module or different modules, respectively, enabling a decomposition of
	directionality into local and global components.
\subsection{Non-normality quantification and perturbation recovery}
\label{sec:methods:nonnorm}

	\paragraph{Relative commutator normality defect.}
	We quantify the non-normality of the recurrent weight matrix using both the
	raw commutator defect and its scale-invariant normalization,
	\begin{equation}\label{eq:dep}
	\begin{aligned}
		\Delta_{\mathrm{comm}}(\mathbf{W})
		&\triangleq \|\mathbf{W}\mathbf{W}^\top-
		\mathbf{W}^\top\mathbf{W}\|_F,\\
		\delta_{\mathrm{rel}}(\mathbf{W})
		&\triangleq
		\frac{\Delta_{\mathrm{comm}}(\mathbf{W})}{\|\mathbf{W}\|_F^2}.
	\end{aligned}
	\end{equation}
where $\|\cdot\|_F$ is the Frobenius norm. For a normal matrix
	($\mathbf{W}\mathbf{W}^\top = \mathbf{W}^\top\mathbf{W}$), both quantities
	vanish; conversely, either zero defect implies normality for nonzero $\mathbf{W}$.
Normal matrices admit an orthonormal eigenbasis over $\mathbb{C}$ and therefore
have no eigenvector-induced transient amplification in a stable linear system.
	A larger defect indicates a larger violation of normality and
can permit transient amplification~\citep{trefethen2005spectra}. This
	commutator defect should not be confused with Henrici's classical Frobenius
departure $\bigl(\|\mathbf{W}\|_F^2-\sum_i|\lambda_i|^2\bigr)^{1/2}$; both
vanish exactly for normal matrices but have different scales
	\citep{henrici1962bounds}. The raw quantity is scale-sensitive,
	$\Delta_{\mathrm{comm}}(c\mathbf{W})=c^2
	\Delta_{\mathrm{comm}}(\mathbf{W})$, whereas
	$\delta_{\mathrm{rel}}(c\mathbf{W})=\delta_{\mathrm{rel}}(\mathbf{W})$ for
	$c\neq0$. We therefore use $\delta_{\mathrm{rel}}$ as the primary
	non-normality diagnostic and retain $\Delta_{\mathrm{comm}}$ and
	$\|\mathbf{W}\|_F^2$ to audit the effect of weight scale. All three are
	computed for each generated network across the full $\rho$ sweep ($n=30$
	trials per $\rho$ level). For the linearised dynamics
	$\mathbf{A}=(1-\alpha)\mathbf{I}+\alpha\mathbf{W}$, the raw identity
	$\Delta_{\mathrm{comm}}(\mathbf{A})=\alpha^2
	\Delta_{\mathrm{comm}}(\mathbf{W})$ holds to machine precision; no analogous
	equality is assumed for the relative defect because its denominator changes.
	For each pairwise mechanism analysis we report Pearson correlations both over
	all independently generated networks and over the 11 $\rho$-level means. To
	distinguish an across-sweep association from residual network-to-network
	variation, we additionally subtract the corresponding level mean from each
	observation and correlate these within-$\rho$ residuals. All reported
	correlation tests are two-sided.

\paragraph{Perturbation recovery experiment.}
To test whether non-normality covaries with, and may index, dynamical fragility,
we measure perturbation recovery speed as a function of $\rho$. For each network, we
	first run the dynamics (Eq.~\ref{eq:esn}) for 2,000 steps with scalar input
	$u(t)=0$; the randomly drawn bias column of $\mathbf{W}_{\mathrm{in}}$ remains
	active and produces a heterogeneous steady state. We then inject a
	perturbation by adding Gaussian noise with standard deviation $5\sigma$ to
	20\% of randomly selected neurons, where $\sigma$ is the across-neuron
	standard deviation of the unperturbed steady state. Both trajectories
continue under identical input. We track the normalised L2 distance
$d(t) = \|\mathbf{x}_{\mathrm{pert}}(t) -
\mathbf{x}_{\mathrm{ref}}(t)\| / \sqrt{N}$ and record the number of steps
	required for $d(t)$ to fall below 10\% of its initial value, with a maximum
	follow-up of 2,000 steps. Five perturbations are applied per network and their
	recovery times are averaged ($n=10$ networks per $\rho$ level); all trials
	recovered before the maximum follow-up.
Shorter recovery steps indicate greater dynamical stability.

	\subsection{Shared-matrix functional-triad analysis}
	\label{sec:methods:functional_triad}

	To measure memory capacity, local restoration, and task-state separability on
	identical recurrent matrices, we reused the 210-network bank from the audited
	post-cue DMS scan: 21 reciprocity levels
	($\rho=0,0.05,\ldots,1$), ten matrices per level, $N=300$, $c=0.1$,
	$\alpha=0.2$, and maximum-real normalization at one. Each matrix retained the
	same $(\rho,\text{network index})$ identifier used in the DMS analysis. Memory
	capacity and restoration were computed on the corresponding recurrent matrix,
	whereas cue margin was taken from the corresponding clean DMS record. The
	three measurements were merged by this shared network identifier.

	\paragraph{\texorpdfstring{\newrev{Online memory capacity $C$, restoration score $S$, and cue margin $D$.}}{Online memory capacity C, restoration score S, and cue margin D.}}
	Memory capacity was measured with the standard online reconstruction assay
	described above. A 2,500-step input $u(t)\sim\mathcal{U}(-0.8,0.8)$ was used
	with a 500-step washout; ridge readouts ($\lambda=10^{-4}$) were trained on
	60\% of the remaining samples to reconstruct delays $1$--$20$ and tested on
	the held-out 40\%. \newrev{We define $C=\sum_{k=1}^{20}\mathrm{MC}_k$ as the
	online memory-capacity coordinate of the functional triad. It is a specific
	driven-task score, distinct from the SSV used to characterize representational
	geometry in the theoretical analysis.} Input
	sequences and input weights were paired across $\rho$ within each network-index
	block.

	For autonomous-state discovery, 32 initial states drawn
	uniformly from $[-1,1]^N$ evolved for 1,000 steps with zero external input and
	a fixed network-index-specific bias. An endpoint was accepted when its final
	per-neuron speed was below $10^{-5}$, and accepted endpoints were clustered
	with a pre-specified per-neuron radius of $0.02$. Each discovered state received
	eight independent impulses applied to 20\% of neurons at amplitudes of one,
	two, and five times the across-neuron activity SD, followed for 300 steps.
	For a distance trajectory $d_t$, recovery AUC is the time average of
	$d_t/d_0$; the restoration score is $S=1-\mathrm{AUC}$ at the strongest
	pre-specified impulse. This score operationalizes dynamical stability as the
	ability to return toward an autonomous state after perturbation. Multiple
	reproducible basins would additionally be required to identify the assay as
	multi-attractor memory.

	Task-state separability was measured in a clean four-class DMS assay. Each cue
	was presented for five steps at amplitude $0.3$ and followed by a 200-step
	silent delay. We define $D$ as the minimum Euclidean distance between the four
	terminal cue centroids, divided by $\sqrt{N}$. Cue labels and input weights
	were paired across $\rho$ within each network-index block. For structural
	interpretation only, we also computed distinct-partner coverage $U$ from the
	pre-balance excitatory support using Eq.~\ref{eq:partner_coverage}; $U$ is not
	a coordinate of the functional triad.

	\paragraph{Frozen-readout noisy evidence integration.}
	Each binary trial had latent class $s\in\{-1,+1\}$ and a 200-step evidence
	epoch $u_t=0.1s+\sigma\epsilon_t$, followed by a 100-step silent hold. A
	shrinkage-LDA readout was trained once at $\sigma=0.5$ using 100 trials per
	class and frozen while testing $\sigma\in\{0,0.25,0.5,1.0,1.5\}$ on 100 new
	trials per class. For the perturbation test, an impulse was applied to 20\% of
	neurons halfway through the silent hold, so no later evidence could rewrite
	the state. An independent, $\rho$-blinded pilot selected the perturbation amplitude
	whose pooled accuracy was closest to 0.75; this froze the confirmatory dose at
	$0.025$ times activity SD. Task-performance retention was
	$(a_{\mathrm{pert}}-0.5)/(a_{\mathrm{ref}}-0.5)$, and the readout remained
	frozen for every test condition.

	\paragraph{Functional-triad visualization and statistical analysis.}
	The recurrent network, not the trial, was the independent replicate. Primary
	trend tests used network-level Spearman correlations and Pearson correlations
	of the 21 $\rho$-level means. As a paired analysis, a linear slope across
	$\rho$ was estimated within each of the ten network-index blocks and tested
	against zero by two-sided one-sample $t$ and Wilcoxon signed-rank tests. Figure
	summaries show mean $\pm$ SD. The 21 level means of $C$, $S$, and $D$ were
	min--max normalized separately only for the overlay in
	Fig.~\ref{fig:functional_triad}d; panels a--c and e retain native
	units. Linear fits to the level means emphasize the direction of each trend;
	the observed means remain visible throughout. The overlay is a descriptive
	summary of three measured functional coordinates, not evidence for Pareto
	optimality or an intermediate optimum.
	All network-level measurements were retained with their matrix and seed
	identifiers.

		\subsection{\texorpdfstring{\rev{Component-matching validation}}{Component-matching validation}}
	\label{sec:methods:csd_matching}

		\rev{We used the same deterministic 210-matrix bank to test whether the stability
		and separability components could be distinguished by controlled matching
		rather than covariance alone.} For each network-index block, we defined a common cue
	margin target from the median natural DMS margin across $\rho$. A separate
	feasibility pass evaluated 15 cue amplitudes logarithmically spaced from
	$0.005$ to $1.5$; the confirmatory target was capped at 90\% of the smallest
	maximum-attainable margin across the 21 matrices in that block. Cue amplitude
	was then selected by bracketed geometric bisection to match this target. The
	realized margins and matching deviations were retained for quality control.

	At the reference process-noise dose $\sigma_x=2.5\times10^{-4}$, we next
	defined a block-specific within-cue dispersion target as the median dispersion
	pooled across natural- and $D$-matched conditions. Noise amplitude was rescaled
	as $\sigma_x\sqrt{V_{\mathrm{target}}/V_{\mathrm{observed}}}$, with bounds of
	$\sigma_x/4$ and $4\sigma_x$. We compared four conditions: natural $D$ and
	$V$; matched $D$ with natural $V$ variation; natural $D$ with matched $V$;
	and both $D$ and $V$ matched. Realized dispersions and matching deviations
	were retained for quality control.

	Each condition was tested with 200 balanced four-class trials at noise
	multipliers $0$, $0.5$, $1$, and $2$ after a 200-step cue--delay interval.
		\rev{Classification used clean cue centroids; no decoder was adapted to noisy
		test trials, and the centroid rule was held fixed across noise doses within
		each matching condition.} Accuracy was converted to
	chance-normalized retention and integrated over noise dose to obtain a
		robustness AUC. Within each network-index block, we fitted the AUC slope
		$\beta$ across all 21 $\rho$ levels. \rev{We defined the operational
		stability-associated contrast as
		$\Delta\beta_S=\beta_{D\text{-matched}}-\beta_{D,V\text{-matched}}$ and the
		separability-associated contrast as
		$\Delta\beta_D=\beta_{V\text{-matched}}-\beta_{D,V\text{-matched}}$.}
		Block-level slopes and paired contrasts were assessed with two-sided
		one-sample $t$ tests and Wilcoxon signed-rank tests. \rev{These contrasts
		quantify changes under the specified amplitude-matching protocol. Because the
		protocol adjusts cue and process-noise amplitudes, it does not constitute a
		general intervention on reciprocity or a complete causal-mediation analysis.}

	As an independent driven-task control on the same matrices, we evaluated a
	held-out NARMA-10 sequence of length 3,200 after a 500-step washout. A ridge
	readout ($\lambda=10^{-4}$) was trained on 60\% of the usable sequence and
	tested on the remaining 40\%; input sequences and input weights were paired
	across $\rho$ within network-index blocks. Performance was reported as
	$1-\mathrm{NMSE}$.

	\subsection{\texorpdfstring{\newrev{Biological connectome and microcircuit datasets}}{Biological connectome and microcircuit datasets}}
	\label{sec:methods:bio}
	
	To provide structural context for the synthetic-network results, we compiled
	\newrev{synapse-level connectome and microcircuit datasets} from five species, \textit{C. elegans},
	\textit{P. dumerilii}, \textit{D. melanogaster}, \textit{C. intestinalis}, and
	\textit{Homo sapiens}. The datasets included the following: the \textit{Ciona
		intestinalis} (ascidian tadpole larva) connectome collected by
	Ryan et al.\ (2016)~\citep{ryan2016cns}; the adult hermaphrodite
	\textit{C. elegans} chemical-synapse connectome from Cook et
	al.\ (2019)~\citep{cook2019whole}; the first-instar
\textit{Drosophila} larval connectome from Winding et
	al.\ (2023)~\citep{winding2023connectome}; the adult \textit{Drosophila melanogaster}
	connectome from Lin et al.\ (2024)~\citep{lin2024network}; and the 3-day larval
	connectome of the marine annelid \textit{Platynereis dumerilii} from
	\newrev{Veraszt\'{o}} et al.\ (2025)~\citep{veraszto2025whole}; and the human temporal
	cortex L2/3 microcircuit from Peng et al.\ (2024)~\citep{peng2024directed}, comprising
	a primary dataset from temporal lobe epilepsy resections (Human TLE, 177
	clusters, represented here as 1,770 cells after block-diagonal cluster union)
	and a validation dataset from tumor surgeries (Human Tumor, 54 clusters,
	540 cells). All cross-species structural comparisons are performed at the
	identified neuron or recorded-cell level rather than at a brain-region level.
	For each non-human connectome, directed adjacency matrices were constructed by
	counting synapses between identified neurons. The human datasets are
	cell-level microcircuit reconstructions within recorded clusters; we analyse
	their block-diagonal union, so unobserved cross-cluster edges are treated as
	absent and the resulting values should be read as within-sampled-microcircuit
	quantities rather than a complete cortical connectome. Self-connections were
	excluded and the resulting matrices were treated as the recurrent weight
	matrix for subsequent topological and spectral analyses.
	Reciprocity follows the standard convention
	$\hat\rho = (\sum_{ij}A_{ij}A_{ji})/\sum_{ij}A_{ij}$ (fraction of directed
	edges that are reciprocated), and reciprocity/cyclicity over-representation
	is reported relative to a directed Erd\H{o}s--R\'{e}nyi null of matched size
	and density.
	
		\paragraph{Connectome state-space volume.}
		An exploratory SSV estimate is obtained by spectral-radius normalizing the
		\emph{actual} connectome, using it as the recurrent matrix $\mathbf{W}$ of a
		leaky-tanh reservoir (Eq.~\ref{eq:esn}), and driving it with random input. We
	then apply Eq.~\ref{eq:ssv} to the resulting activity, so that the measure
	reflects the real directed topology (reciprocity, cyclicity, degree
		structure) rather than only network size and density. The current pipeline uses
		one fixed random seed. Because absolute SSV scales with $N$ and sequence length,
		we report $\mathrm{SSV}/N$ descriptively, but this division does not remove
		finite-sample $N/T_{\mathrm{eff}}$ bias. Pipeline-produced finite values are
		reported for all seven datasets. For connectomes with
		$N\geq T_{\mathrm{eff}}$, these values depend on the numerical eigenvalue floor
		used to evaluate the log determinant and should not be interpreted as
		unregularized SSV estimates. We therefore make no inferential cross-species
		comparison of SSV pending a shrinkage or matched-subgraph analysis. The cross-species topological comparison
	(Fig.~\ref{fig:species}) uses the seven datasets for which a directed
	adjacency matrix could be reconstructed and verified: the \textit{Ciona}
	larva, Cook \textit{C. elegans} (hermaphrodite chemical), \textit{P. dumerilii}
	larva, Winding \textit{Drosophila} larva, FlyWire adult \textit{Drosophila},
	and the two human microcircuits.
	
	For \textit{Drosophila}, region-level directed adjacency matrices are partitioned
	into anatomical regions using established atlas parcellations; for each
	deposited region matrix, we compute reciprocity, modularity $Q$, and
	within- and between-module reciprocity ($\rho_{\mathrm{in}}$,
	$\rho_{\mathrm{out}}$). 
	
	\section{Results}
	\label{sec:results}
	
	\subsection{Increasing reciprocity reduces state-space volume and driven-task performance}
	\label{sec:results:main}
	
	We systematically varied $\rho$ from 0 to 1 across 21 levels and measured
	the state-space volume (SSV) of each reservoir together with its performance
	on four driven transient-processing tasks (Fig.~\ref{fig:rho_sweep}).
	
	State-space volume decreased monotonically with increasing reciprocity
	(Fig.~\ref{fig:rho_sweep}a). At $\rho = 0$, mean SSV was
	$428.6 \pm 54.4$ nats; at $\rho = 1$, it fell to $-151.4 \pm 26.1$ nats
	($r = -0.992$,  $p < 0.001$). SSV crossed zero near $\rho \approx 0.64$,
	but this sign change has no separate dynamical interpretation because the
	absolute log-volume depends on sequence length and scale (Methods). The
	relevant result is the monotonic comparison at fixed $N$ and
	$T_{\mathrm{eff}}$.
	
		All four tasks mirrored the SSV trend. Spoken Arabic Digit classification
		accuracy declined from $0.9404 \pm 0.0053$ at $\rho = 0$ to $0.9195 \pm
		0.0046$ at $\rho = 1$ ($r = -0.963$, $p < 0.001$; Fig.~\ref{fig:rho_sweep}b).
		Mackey--Glass forecasting error rose from NMSE $= 0.0043 \pm 0.0014$ to
		$0.0713 \pm 0.0097$ ($r = +0.924$, $p < 0.001$; Fig.~\ref{fig:rho_sweep}c).
		Memory capacity dropped from $5.22 \pm 0.62$ to $4.13 \pm 0.34$ ($r = -0.880$,
		$p < 0.001$; Fig.~\ref{fig:rho_sweep}d). Multi-sine oscillator prediction
		error increased from NMSE $= 0.1668 \pm 0.0209$ to $0.4613 \pm 0.0909$
		($r = +0.959$, $p < 0.001$; Fig.~\ref{fig:rho_sweep}e).
	
	\begin{figure}[htbp!]
		\centering
		\includegraphics[width=\textwidth]{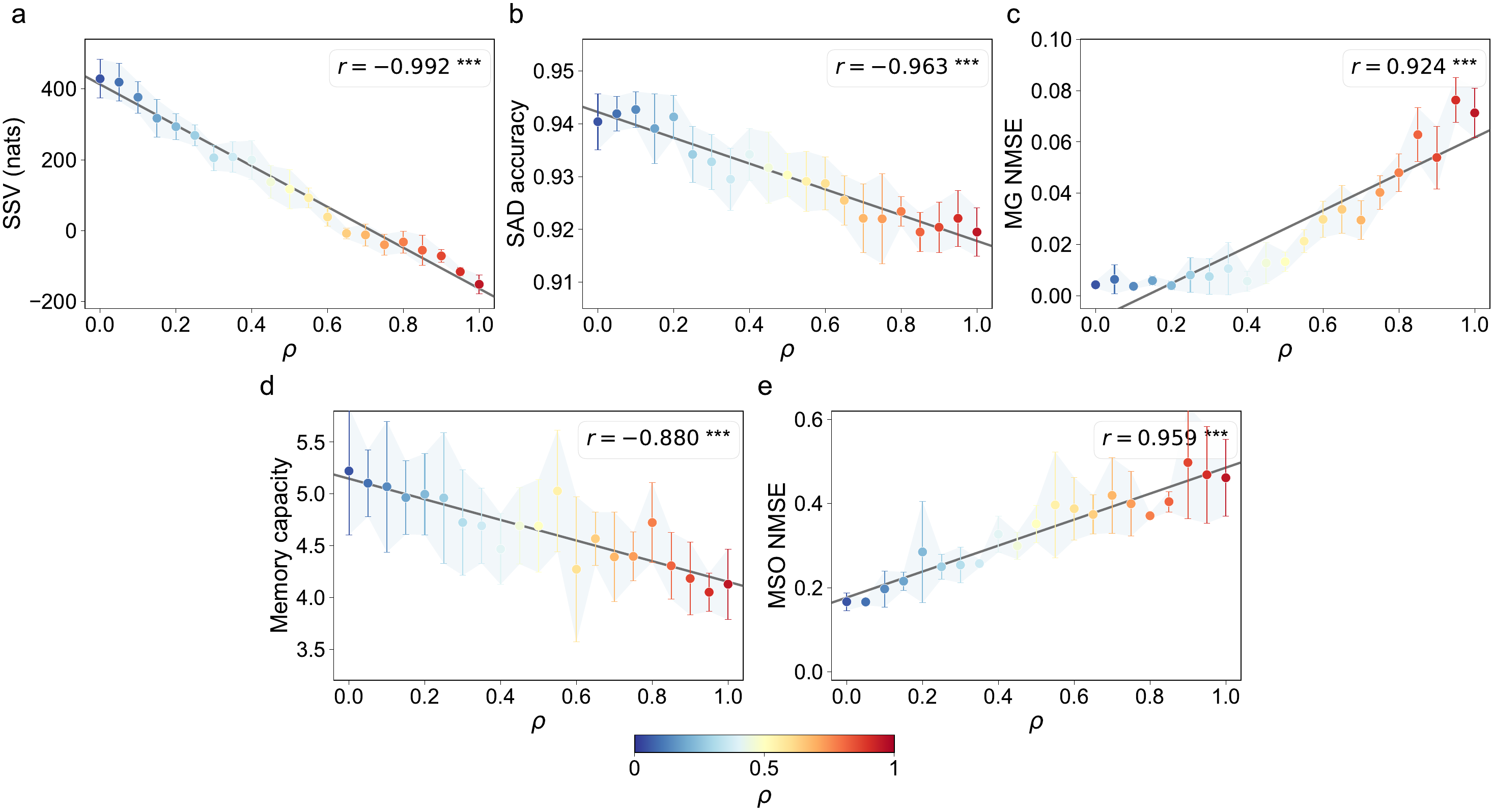}
		\caption{\textbf{Increasing reciprocity reduces representational volume and
				performance on four driven tasks.} \textbf{(a)} State-space volume (SSV). \textbf{(b)}
			Spoken Arabic Digit classification accuracy. \textbf{(c)} Mackey--Glass
			forecasting NMSE. \textbf{(d)} Memory capacity. \textbf{(e)} Multi-sine
		oscillator prediction NMSE. Markers: mean $\pm$ SD ($n = 10$
		independently generated networks per $\rho$); black line: linear fit;
		Pearson $r$ reported in each panel.
			All panels: $N = 300$, $c = 0.1$, maximum-real target
			$\sigma = 1.0$.}
		\label{fig:rho_sweep}
	\end{figure}
	
	\subsection{\texorpdfstring{\newrev{Biological connectivity datasets occupy an intermediate-reciprocity regime}}{Biological connectivity datasets occupy an intermediate-reciprocity regime}}
	\label{sec:results:bio}
	
	We analysed the seven \newrev{connectome and microcircuit datasets} described in Methods
	(Fig.~\ref{fig:species}) using a uniform computational pipeline to position
	them on the reciprocity axis explored in our synthetic networks. We treat this
	as a descriptive structural comparison, not as an independent functional test
	of the proposed three-way dynamical relationship, because SSV is not uniformly
	comparable across the seven datasets.
	
	\newrev{The biological connectivity datasets} spanned a broad band of the reciprocity axis rather
	than clustering at the unidirectional extreme: unidirectionality ranged from
	$0.41$ (\textit{Ciona}, the most reciprocal) through $0.72$--$0.79$
	(\textit{C. elegans} hermaphrodite $0.72$, \textit{Drosophila} larva $0.74$,
	Human TLE $0.73$, Human Tumor $0.79$) and $0.83$ for the adult
	\textit{Drosophila} whole brain to $0.88$ for \textit{P. dumerilii} (the most
	unidirectional), i.e.\ realized reciprocity $\hat\rho \approx 0.12$--$0.59$.
	Despite this intermediate \emph{absolute} reciprocity, reciprocal
	motifs were strongly over-represented relative to ER in every dataset
	($8.3\times$ in \textit{Ciona} to $1.68\times10^{3}\times$ in the adult
	\textit{Drosophila} brain). Cyclicity, quantified from directed 3-cycles, was
	also strongly over-represented
	($6\times$ to $1.5\times10^{3}\times$ ER), reflecting that sparse, spatially
	embedded connectomes have a very low random baseline rather than near-perfect
		directedness \citep{hoffmann2017nonrandom}. Single-fixed-seed per-neuron SSV
		values were $-6.747$ for \textit{C. elegans}, $-9.462$ for
		\textit{P. dumerilii}, $-8.429$ for \textit{Ciona}, $-9.760$ for larval
		\textit{Drosophila}, $-11.23$ for adult \textit{Drosophila}, $-9.378$ for
		Human TLE, and $-8.769$ for Human Tumor. These values remain descriptive
		because finite-sample bias varies with $N/T_{\mathrm{eff}}$, and values from
		rank-deficient cases depend on the numerical eigenvalue floor. We therefore
		make no cross-species inferential claim from SSV. The topological
		reciprocity comparisons above are unaffected by this estimator limitation.
	
		\begin{figure}[htbp]
			\centering
				\includegraphics[width=\textwidth]{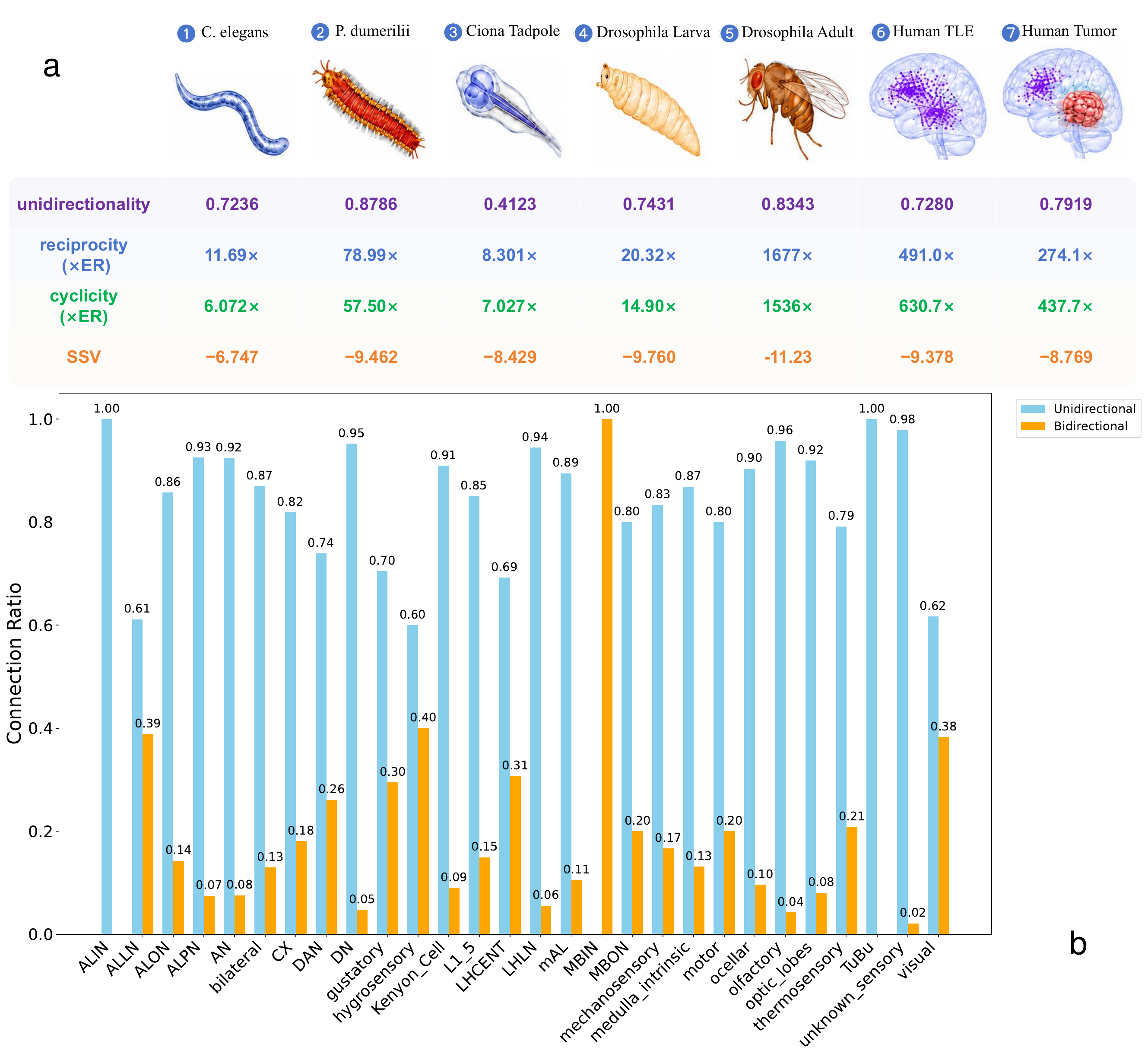}
		\caption{\textbf{\newrev{Biological connectivity datasets} span an intermediate-reciprocity
				regime with substantial within-brain directionality variation.}
			\textbf{(a)} \newrev{Cross-species connectome and microcircuit topology.} The illustrated table reports
			unidirectionality ($\upsilon = 1-\hat\rho$), reciprocity relative to ER,
	and cyclicity relative to ER across seven \newrev{connectome and microcircuit
	datasets (five species). Human entries are block-diagonal unions of sampled
	temporal-cortex L2/3 microcircuits.} Cyclicity is quantified from directed 3-cycles. Exploratory
	per-neuron SSV is shown for all seven datasets.
		\textit{P. dumerilii} is the most unidirectional
		and \textit{Ciona} the most reciprocal. SSV is a single-fixed-seed
		Monte Carlo estimate; values from rank-deficient cases depend on the
		numerical eigenvalue floor. The reciprocity
			over-representation factors for \textit{C. elegans}, the
			\textit{Drosophila} larva, and the adult \textit{Drosophila} brain are
			consistent with those reported by \citet{lin2024network}.
		\textbf{(b)} Directionality gradient across \textit{Drosophila} brain
		regions. Unidirectional and bidirectional fractions of intra-region
		connectivity for 28 anatomical regions, ordered from strongly unidirectional
		to strongly bidirectional. }
		\label{fig:species}
	\end{figure}
	
		Across these evolutionarily distant nervous systems the pattern is therefore
		not one of maximal unidirectionality but of \emph{intermediate} reciprocity
		combined with strong, density-driven over-representation of reciprocal and
		cyclic motifs. Their location overlaps the interior of our synthetic sweep,
		although this structural comparison alone cannot assign a connectome to an
		encoding- or stability-favouring computational regime.
	
	To examine whether reciprocity varies systematically
	\emph{within} a single connectome, we partitioned the
	\textit{Drosophila} adult brain into anatomical regions and computed
	unidirectionality for each region (Fig.~\ref{fig:species}b). A
		broad gradient emerged. ALIN and TuBu contained only unidirectional edges
		($\upsilon=1.00$), while unknown sensory, olfactory, and descending-neuron
		groups were also strongly unidirectional ($\upsilon=0.98$, $0.96$, and
		$0.95$, respectively). Other deposited groups contained larger
		bidirectional fractions, including visual ($0.38$ bidirectional), ALLN
		($0.39$), and hygrosensory ($0.40$). One small deposited group was fully
		bidirectional, illustrating why extreme regional ratios should not be
		interpreted without accounting for node and edge counts. More
		generally, the gradient does not separate cleanly by sensory or integrative
		label, and region size and sampling vary widely. We therefore treat the
		within-brain analysis as evidence of anatomical heterogeneity and as a
		hypothesis-generating comparison, rather than as a statistical test that
		functional role determines reciprocity~\citep{zhang2025network,zhang2025structural}.

\subsection{Seeded replication supports task-dependent delayed-memory benefits}
\label{sec:results:tradeoff}

The driven-task results raise a natural question: why do biological connectomes
retain reciprocal connections if low reciprocity benefits every task in the
original battery? We hypothesised that reciprocity can strengthen local
contraction and noise attenuation, but that the behavioural benefit of this
stability should depend on whether cue-conditioned states remain separable
through long delays or intervening distractors.

To test this, we evaluated an independently seeded ensemble generated with the
same network protocol on the three delayed-memory tasks described in Methods
(Adding Problem, Delayed Match-to-Sample, Pattern Hold).
All require selected information to survive a substantial interval during
which it is no longer re-presented: Adding contains ongoing distractors, DMS
contains noise but no informative cue, and Pattern Hold contains a silent
delay. This distinguishes delay-period maintenance from the continuously driven
trace-memory protocol used for MC.

The seeded results show an aggregate reversal relative to the driven tasks
(Fig.~\ref{fig:tradeoff}). For the representative $T=200$ Adding condition,
NMSE decreased from $0.961 \pm 0.115$ at $\rho=0$ to $0.740 \pm 0.073$ at
$\rho=1$ (Pearson correlation of the 21 $\rho$-level means, $r=-0.856$,
$p=7.7\times10^{-7}$). DMS accuracy at a 200-step delay increased from
$0.814 \pm 0.078$ to $0.980 \pm 0.029$ ($r=0.874$,
$p=2.2\times10^{-7}$). At 400 steps the increase was smaller, from
$0.265 \pm 0.032$ to $0.335 \pm 0.054$ ($r=0.692$,
$p=5.1\times10^{-4}$), and remained close to four-class chance ($0.25$).
Pattern Hold NMSE at a 50-step delay decreased from
$0.00966 \pm 0.00104$ to $0.00575 \pm 0.00042$ ($r=-0.944$,
$p=1.3\times10^{-10}$). All seven reported task--delay conditions had a
significant overall trend in the predicted direction. For each condition,
all ten independent seed-block slopes had the expected sign and remained
significant after correction ($q\leq0.0055$). None of the seven
$\rho$-level mean curves was pointwise strictly monotonic; DMS at 100 steps was
constant at ceiling accuracy and is omitted. These data therefore support
task- and delay-dependent aggregate benefits of higher reciprocity, not a
universal monotonic law for every network or delay.

\begin{figure}[htbp]
	\centering
	\includegraphics[width=\textwidth]{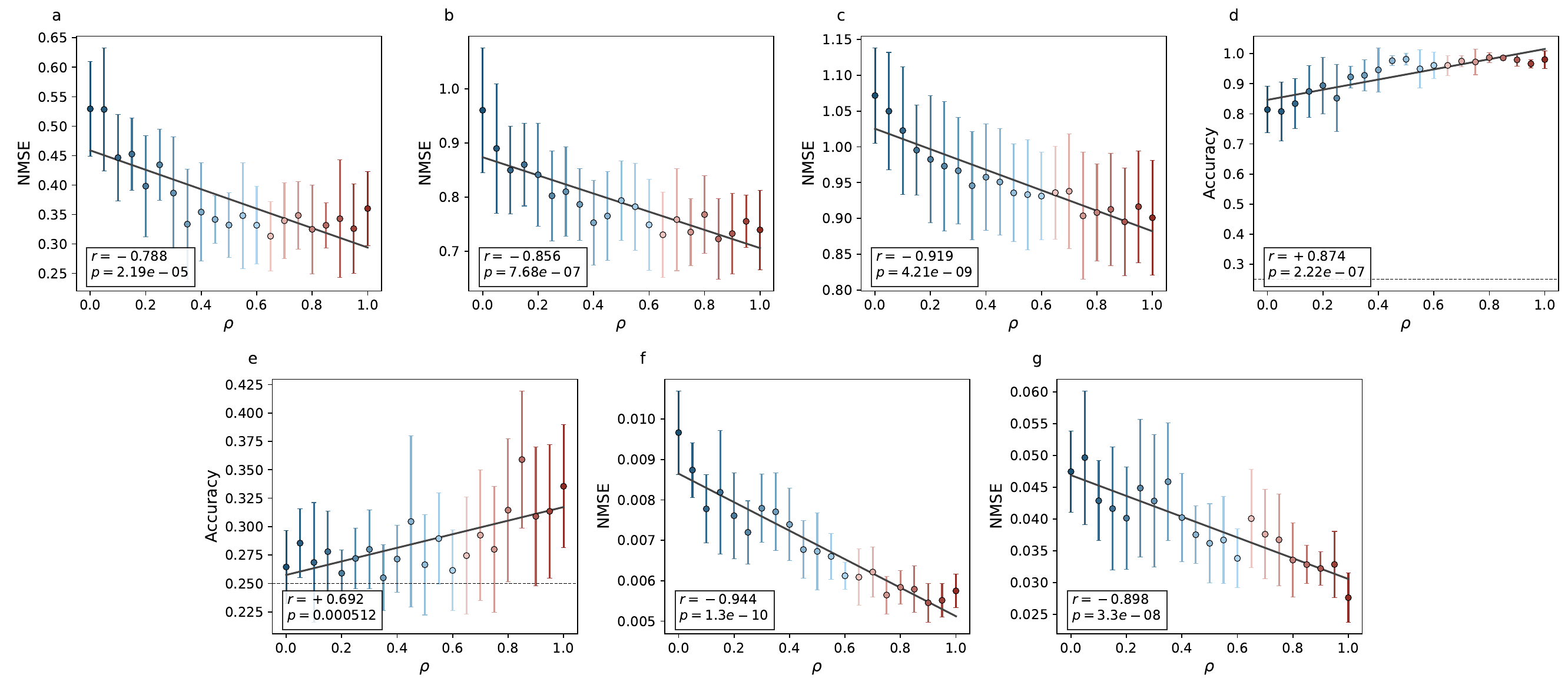}
	\caption{\textbf{Delayed-memory tasks show aggregate trends opposite to the
		driven-task sweep.}
		All seven reported task--delay conditions show significant correlations of
		the $\rho$-level means in the high-reciprocity-benefit direction; no curve is
		pointwise strictly monotonic. \textbf{(a--c)} Adding Problem at
		three sequence lengths ($T = 100$, $200$, $400$); NMSE decreases as $\rho$
		increases. \textbf{(d,e)} Delayed
		Match-to-Sample at two delay lengths ($200$, $400$ steps); accuracy
		shows an overall increase with $\rho$; the omitted 100-step condition is at
		ceiling for all levels.
		\textbf{(f,g)} Pattern Hold at $50$ and $100$ steps. Each panel: mean $\pm$ SD ($n = 10$ independently generated networks
		per $\rho$); grey line is a linear fit to the 21 level means; Pearson $r$ and
		$p$ are reported in each sub-panel.
		}\label{fig:tradeoff}
\end{figure}

\subsection{Post-cue noise exposes the separability constraint on behavioural robustness}
\label{sec:results:dms_process_noise}

To distinguish delay retention from noise resistance, we injected isotropic
state noise only after cue presentation in the same 210 networks
(Fig.~\ref{fig:dms_process_noise}). At $\sigma_x=10^{-4}$, empirical noise gain
decreased from $12.25\pm3.88$ at $\rho=0$ to $4.79\pm0.35$ at $\rho=1$
($r_s=-0.863$, $p=1.7\times10^{-63}$). Higher reciprocity therefore attenuated
post-cue state noise in this operating regime.

Noise attenuation alone did not improve endpoint accuracy. Frozen-readout AUC
changed only weakly with $\rho$ ($r_s=0.226$, $p=9.7\times10^{-4}$), and its
paired endpoint difference included zero; the noise-matched readout showed no
	\rev{monotonic trend ($r_s=0.077$, $p=0.264$).} Meanwhile, the clean cue margin declined across
$\rho$ levels ($r=-0.575$, $p=0.0064$). The combined proxy
$m_{\mathrm{clean}}^2/G_L^{\mathrm{emp}}$ was strongly associated with
frozen-readout robustness ($r_s=0.722$, $p=4.2\times10^{-35}$;
Fig.~\ref{fig:dms_process_noise}f), and all ten seed-block correlations were
positive (range $0.330$--$0.869$; exact sign test $p=0.002$). Behavioural
robustness thus depends on noise attenuation relative to the cue separation
retained by the network.

\begin{figure}[htbp]
	\centering
	\includegraphics[width=\textwidth]{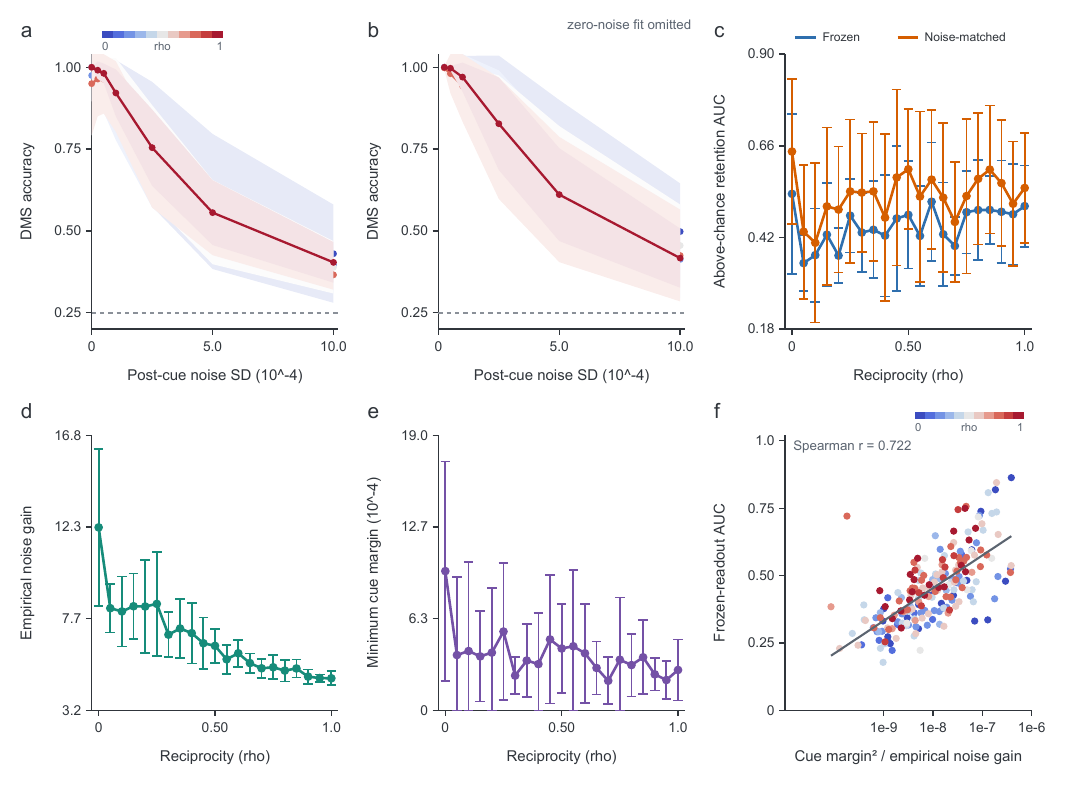}
	\caption{\textbf{Post-cue noise reveals a stability--separability boundary
		for behavioural robustness.}
		\textbf{(a)} Four-class accuracy of the readout fitted once to noiseless
		states and frozen across the dose curve. \textbf{(b)} Accuracy when the
		linear readout is refitted independently at each positive noise dose; the
		exactly zero-noise point is omitted because its noise-matched covariance
		estimate is degenerate.
		\textbf{(c)} Above-chance retention AUC for frozen and noise-matched readouts
		across all 21 reciprocity levels.
		\textbf{(d)} Paired empirical terminal noise gain at
		$\sigma_x=10^{-4}$; lower values indicate stronger attenuation.
		\textbf{(e)} Minimum centroid margin between clean cue-conditioned terminal
		states. \textbf{(f)} Frozen-readout AUC versus the measured margin/noise
		proxy $m_{\mathrm{clean}}^2/G_L^{\mathrm{emp}}$; each point is one network
		and colour denotes $\rho$. Curve colour in (a,b) likewise denotes $\rho$.
		Curves and $\rho$-level summaries show mean
		$\pm$ SD over $n=10$ independently generated networks per level; shaded
		bands in (a,b) are $\pm$ SD, and horizontal dashed lines mark four-class
		chance accuracy.}
	\label{fig:dms_process_noise}
\end{figure}
\FloatBarrier

\subsection{\texorpdfstring{\newrev{Matched networks reveal an online-memory--stability--separability trajectory}}{Matched networks reveal an online-memory--stability--separability trajectory}}
\label{sec:results:functional_triad}

We measured \newrev{online memory capacity} $C$, autonomous-state restoration $S$, and clean
cue margin $D$ on the same 210 recurrent matrices
(Fig.~\ref{fig:functional_triad}). Memory capacity decreased from
$4.317\pm0.358$ at $\rho=0$ to $3.218\pm0.278$ at $\rho=1$
($r=-0.990$, $p=8.8\times10^{-18}$ across level means;
Fig.~\ref{fig:functional_triad}a), and all ten paired seed blocks had negative
slopes (Wilcoxon $p=0.0020$). Restoration increased from
$0.9290\pm0.0100$ to $0.9583\pm0.0054$
($r=0.941$, $p=2.0\times10^{-10}$; Fig.~\ref{fig:functional_triad}b), with
positive slopes in all ten blocks (Wilcoxon $p=0.0020$). All 32 initial conditions
converged to the same nonzero, bias-supported state within each network, so
$S$ measures local restoration rather than cue-addressable multi-attractor
memory.

Cue margin decreased from $(9.63\pm7.57)\times10^{-4}$ at $\rho=0$ to
$(2.79\pm2.11)\times10^{-4}$ at $\rho=1$
($r=-0.575$, $p=0.0064$ across level means;
Fig.~\ref{fig:functional_triad}c). Eight of ten paired seed blocks had negative
slopes, and their slopes differed from zero by both the one-sample $t$ test
($p=0.0372$) and Wilcoxon test ($p=0.0371$). Thus, across the same matrix bank,
increasing reciprocity was accompanied by lower memory capacity, stronger local
restoration, and reduced cue separation. The normalized overlay displays this
joint shift directly (Fig.~\ref{fig:functional_triad}d).

The fixed directed-edge budget supplies a structural interpretation but not a
fourth functional coordinate. Distinct-partner coverage $U$ obeyed the exact
identity $U=c_{\mathrm r}(2-\rho_{\mathrm r})$ and fell from 59.8 to 29.9
partners per neuron between $\rho=0$ and $\rho=1$. Thus increasing reciprocity
reallocated directed edges from new partner pairs to feedback closure in this
ensemble; this algebraic coupling does not by itself establish the cause of
the measured dynamical changes.
	\subsection{\texorpdfstring{\rev{Component matching separates stability- and separability-associated trends}}{Component matching separates stability- and separability-associated trends}}
\label{sec:results:csd_matching}

The shared-matrix trajectory establishes joint variation of $C$, $S$, and $D$,
but does not by itself determine how those components enter behaviour. We
	therefore applied \rev{a controlled component-matching protocol} to the same 210-matrix
bank. \newrev{As an additional driven-task check,} held-out NARMA-10
performance decreased with reciprocity in
all ten paired seed blocks ($t$ test $p=3.4\times10^{-8}$; Wilcoxon
$p=0.0020$). The measured \newrev{online memory capacity} $C$
predicted NARMA-10 performance across networks ($r_s=0.639$,
$p=1.9\times10^{-25}$), and the association remained positive after removing
$\rho$-level means ($r=0.207$, $p=0.0025$).

Cue-margin matching removed the natural decline in $D$
(Fig.~\ref{fig:functional_triad}c); the median relative matching errors
were 0.30\% for $D$ and 1.31\% for within-cue dispersion $V$. With $D$ held constant,
empirical noise gain still decreased with reciprocity in all ten seed blocks
($t$ test $p=1.5\times10^{-6}$; Wilcoxon $p=0.0020$;
Fig.~\ref{fig:functional_triad}e), and frozen-readout robustness increased
in all ten blocks ($p=2.7\times10^{-4}$). Matching within-cue dispersion as well
	attenuated this positive behavioural slope in every block. \rev{The resulting
	stability-associated contrast was positive in all ten blocks}
($\Delta\beta_S=0.0902$, $p=6.6\times10^{-7}$; Wilcoxon $p=0.0020$;
Fig.~\ref{fig:functional_triad}f).

	The separability-associated contrast was negative in seven of ten blocks
	($\Delta\beta_D=-0.103$)\newrev{, but was not independently distinguishable from
	zero ($p=0.136$)}. The stability- and separability-associated contrasts
	differed within the paired blocks (mean difference $0.193$,
	$t$ test $p=0.017$; Wilcoxon $p=0.027$). \reviewrev{Under this specified matching
	protocol, holding cue margin constant exposed a reproducible
	stability-associated gain.}

\begin{figure}[htbp]
	\centering
	\includegraphics[width=\textwidth]{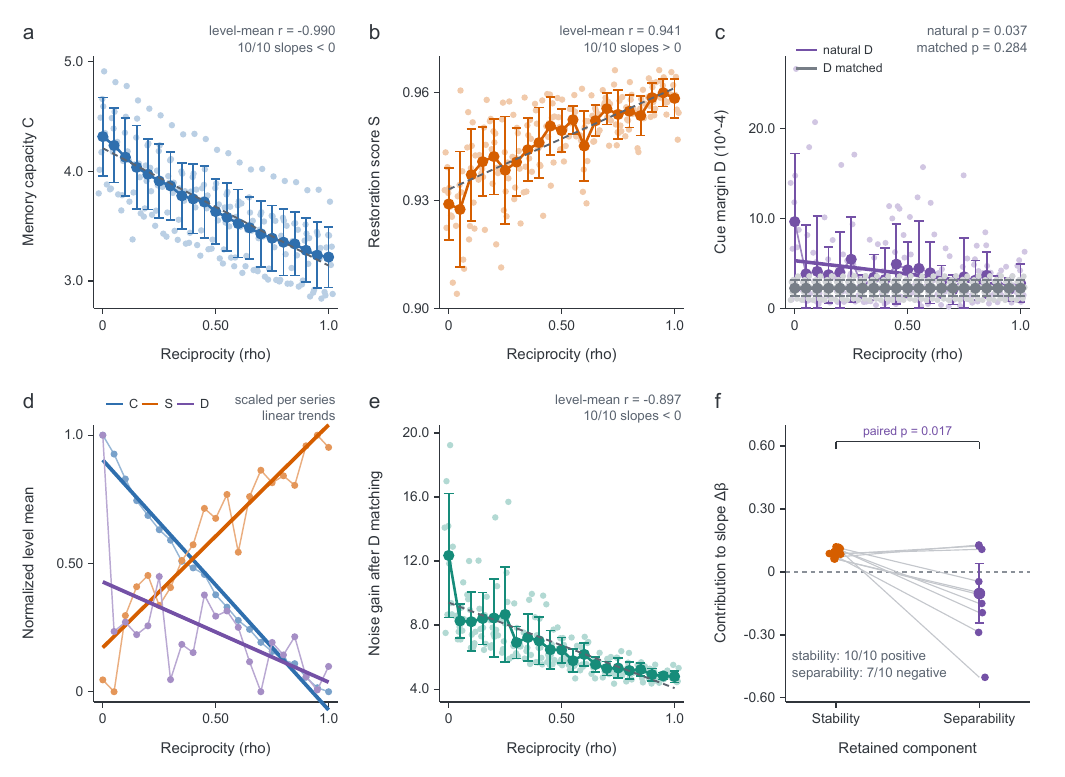}
		\caption{\textbf{\newrev{Matched recurrent matrices reveal an
				online-memory--stability--separability trade-off; component matching compares
				its behavioural components.}}
		\textbf{(a)} Jaeger-style summed memory capacity $C$ over delays 1--20.
		\newrev{This coordinate measures online temporal reconstruction and is distinct
		from representational geometry measured by SSV.}
		\textbf{(b)} Autonomous-state restoration score $S$, defined as
		$1-$normalized recovery-distance AUC after a $5\times$ activity-SD impulse.
		\textbf{(c)} Natural clean DMS cue margin $D$ and cue margin after
		\rev{component matching of $D$.} Natural $D$ is the minimum distance between four
		cue-conditioned terminal-state centroids after a 200-step delay.
		\textbf{(d)} Min--max normalized $\rho$-level means of $C$, $S$, and $D$;
		each series is scaled independently for visualization, and thick lines show
		linear trends through the 21 observed level means. \textbf{(e)} Empirical
		process-noise gain after cue-margin matching. In (a--c,e), pale points are
		individual networks, large points show mean $\pm$ SD ($n=10$ per $\rho$
		level), and dashed or straight lines show trends through level means.
		\textbf{(f)} \rev{Block-level component contrasts for the reciprocity slope}
		$\beta$ of frozen-readout robustness AUC. \rev{The stability-associated contrast is}
		$\Delta\beta_S=\beta_{D\text{-matched}}-\beta_{D,V\text{-matched}}$;
			\rev{the separability-associated contrast is}
		$\Delta\beta_D=\beta_{V\text{-matched}}-\beta_{D,V\text{-matched}}$.
		Small points are paired network-index blocks, lines connect the same block,
		and large points with error bars show mean $\pm$ 95\% CI ($n=10$). The
			annotated $p$ value compares the two \rev{contrasts} with a paired $t$ test.
		Panels (a--f) use the same 210 recurrent matrices across 21 $\rho$ levels.}
	\label{fig:functional_triad}
\end{figure}
\FloatBarrier

\subsection{\texorpdfstring{\newrev{Scale-invariant non-normality tracks response volume and recovery}}{Scale-invariant non-normality tracks response volume and recovery}}
\label{sec:results:mechanism}

We next examined scale-invariant non-normality as a spectral correlate of the
opposing \newrev{SSV} and recovery trends. For the pre-balance excitatory ensemble,
the paired-edge generator obeys
\begin{equation}
    \mathbb{E}\|\mathbf{W}_{e}-\mathbf{W}_{e}^{\top}\|_{F}^{2}
    =2E\bigl[v+(1-\rho)\mu^{2}\bigr],
    \label{eq:rho_antisymmetry}
\end{equation}
where $E$ is the directed-edge count and nonzero weights have mean $\mu$ and
variance $v$ (Supplementary Note~S5). Reciprocity therefore removes the
antisymmetric energy contributed by edge orientation. Because inhibition and
spectral normalization can reintroduce asymmetry, we quantify the final signed
matrix using the relative commutator defect $\delta_{\mathrm{rel}}$ rather than
inferring normality from topology~\citep{shao2023relating,shao2025impact}.

$\delta_{\mathrm{rel}}$ decreased from $0.08346$ at $\rho=0$ to $0.05895$ at
$\rho=1$ ($r=-0.939$, $p=1.9\times10^{-5}$ across level means;
Fig.~\ref{fig:mechanism}a). Its association with SSV was strong across networks
($r=0.873$, $p=2.8\times10^{-104}$; Fig.~\ref{fig:mechanism}c) but weak after
removing $\rho$-level means ($r=0.109$, $p=0.048$). Non-normality is therefore
an ensemble-level correlate, not an established within-level mediator.

For a stable linear reservoir, standardized SSV is the log-determinant of a
row-normalized controllability Gramian, or equivalently the log-volume spanned
by normalized impulse responses (Theorem~\ref{thm:standardized_gramian}). A
fixed-spectrum two-dimensional construction further shows that non-normal shear
can increase SSV and transient gain while reducing delayed-cue information
under noise (Theorem~\ref{thm:shear_standardized}). This trade-off depends on
input alignment and noise; non-normality alone does not order arbitrary tasks.
Consistently, for a stable comparison matrix $\mathbf{A}$, the Kreiss Matrix
Theorem~\citep{trefethen2005spectra} gives
\begin{equation}
	\mathcal{K}(\mathbf{A})\leq
	\sup_{k \geq 0} \|\mathbf{A}^k\|_2
	\leq e N\,\mathcal{K}(\mathbf{A}),
	\label{eq:kreiss}
\end{equation}
so non-normal systems can amplify perturbations beyond predictions from their
eigenvalues alone. Here this theorem serves as a comparison principle, not a
global bound on the nonlinear reservoir.

We tested the prediction by perturbing 20\% of neurons after the network reached
its bias-supported steady state. Recovery steps decreased with reciprocity
($r=-0.955$, $p=4.7\times10^{-6}$ across level means;
Fig.~\ref{fig:mechanism}b). Recovery also covaried with
$\delta_{\mathrm{rel}}$ across networks ($r=0.639$,
$p=6.1\times10^{-14}$; Fig.~\ref{fig:mechanism}d), but not within fixed
$\rho$ levels ($r=0.003$, $p=0.978$). Scale-invariant non-normality therefore
tracks the ensemble shift from greater response volume to faster recovery, but
does not establish a causal mechanism. Figure~\ref{fig:dms_process_noise}
separately shows why faster recovery is useful only when cue-conditioned states
remain distinguishable.

\begin{figure}[!htbp]
	\centering
	\includegraphics[width=\textwidth]{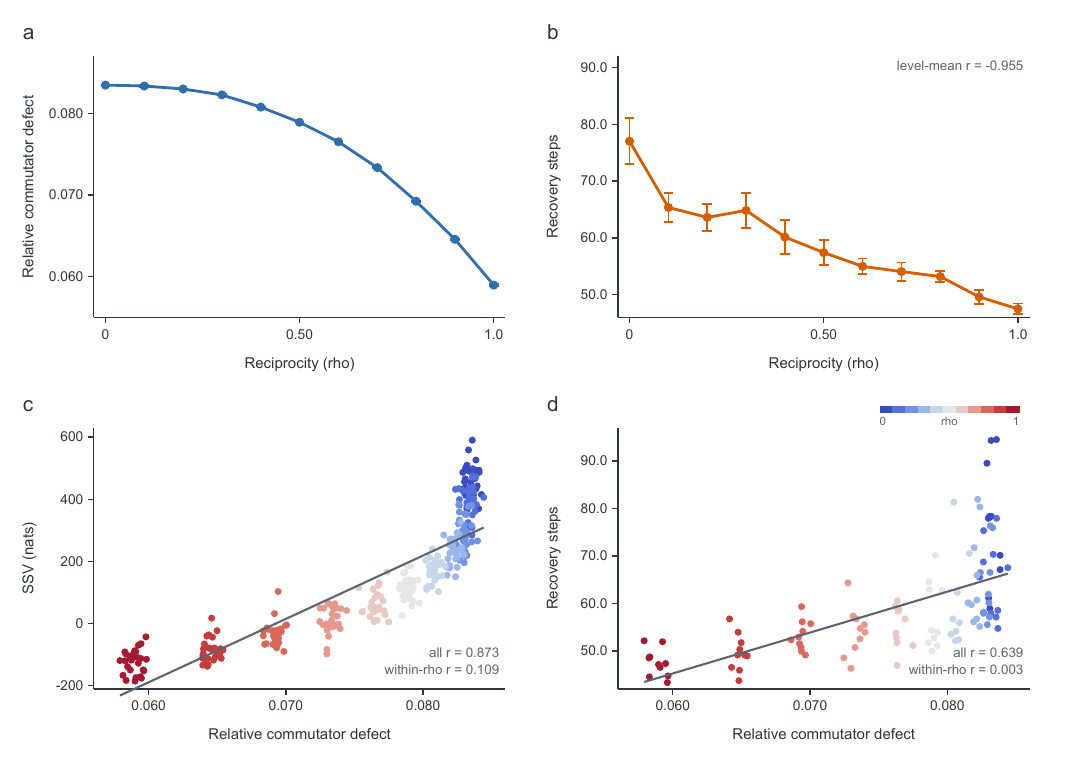}
	\caption{\textbf{Scale-invariant non-normality covaries at the ensemble level
			with response volume and perturbation recovery.}
		\textbf{(a)} Relative commutator normality defect
		$\delta_{\mathrm{rel}}(\mathbf{W})$ versus reciprocity $\rho$.
		\textbf{(b)} Perturbation recovery steps versus $\rho$. Higher
		reciprocity accelerates recovery (level-mean $r=-0.955$,
		$p=4.7\times10^{-6}$).
		\textbf{(c)} SSV versus $\delta_{\mathrm{rel}}$, coloured by
		$\rho$ (all networks $r=0.873$; within-$\rho$ $r=0.109$).
		\textbf{(d)} Recovery steps versus $\delta_{\mathrm{rel}}$,
		coloured by $\rho$ (all networks $r=0.639$; within-$\rho$ $r=0.003$).
		Grey lines in (c,d): fits across all displayed networks. Points and error
		bars in (a,b) show level means $\pm$ SEM; in (a), SEM is smaller than the
		markers. All panels:
		$N = 300$, $c = 0.1$, maximum-real target $\sigma = 1.0$;
		(a,c) 30 trials per $\rho$, (b,d) 10 trials per $\rho$.}
	\label{fig:mechanism}
\end{figure}
\FloatBarrier

Additional control analyses are provided in the Supplementary Materials. The
decrease in state-space volume persisted across network sizes, alternative
topologies, connection densities, and E/I-balance constructions
(Supplementary Notes~S1, S3, and S6), while the effective-rank analysis
provided an independent spectral description of the same dimensionality shift
(Supplementary Note~S4). In contrast, the restoration trend depended on the
operating-point normalization and was absent under the paired maximum-modulus
control; input-noise robustness and perturbed-task retention likewise showed no
universal monotonic dependence on reciprocity (Supplementary Note~S7). These
controls delimit the scope of the main result: reciprocity organizes a
task- and operating-regime-dependent trade-off rather than imposing a
universal performance law.

\section{Discussion}
\label{sec:discussion}

Our central finding is that reciprocity $\rho$ coordinates a three-way,
task- and operating-regime-dependent dynamical trade-off among representational
capacity, local stability, and task-state separability. Low reciprocity
maximises representational volume and driven encoding capacity. Higher
reciprocity strengthens local restoration and reduces internal noise
propagation in the original maximum-real operating regime, while the post-cue
DMS scan shows that this benefit can coincide with smaller margins between
cue-conditioned states. \newrev{The unified 210-network analysis places online
memory capacity $C$, stability $S$, and cue margin $D$ on the same recurrent
matrices: online memory capacity and cue separation decrease as stability
strengthens. Here $C$ quantifies a specific driven temporal task; the broader
representational-geometry result is supported separately by SSV and the driven
task battery.} This matched
measurement identifies a functional trajectory rather than an optimized
Pareto frontier. \rev{The component-matching protocol tests these behavioural
components beyond their natural covariance within this assay: holding cue
margin fixed exposes a reproducible stability-associated gain in frozen-readout
robustness, and matching within-cue dispersion attenuates that gain in every
paired seed block.}

This distinction explains why dynamical stability and behavioural reliability
need not vary together. Increasing $\rho$ strongly reduces within-cue
dispersion by attenuating the propagation of internal state noise, but it also
reduces between-cue margins. Behavioural robustness improves only when the
benefit of reduced noise amplification is sufficient to compensate for the
loss of task-state separation. Consistent with this interpretation,
$m_{\mathrm{clean}}^2/G_L^{\mathrm{emp}}$ is strongly associated with
frozen-readout robustness,
whereas local recovery or reciprocity alone does not yield a monotonic endpoint
advantage.

In the generated ensemble, the relative
commutator normality defect $\delta_{\mathrm{rel}}(\mathbf{W})$ tracks $\rho$
and covaries with SSV and recovery time, but has not been shown to mediate task
performance independently of $\rho$. The formal linear analysis identifies
the relevant encoding object as the normalized controllability Gramian and
proves a fixed-spectrum constructive trade-off: non-normal response spreading
can enlarge standardized SSV while reducing cue information after a specified
noisy delay. A complementary fixed-spectrum theorem shows that normal
dynamics attain the minimum total amplification of isotropic process noise,
while the Gaussian decoding corollary makes cue margin an explicit second
requirement. The nonlinear task ordering remains an empirical ensemble result.

The task evidence for opposing computational demands is direct but conditional.
On four driven tasks, including classification, chaotic forecasting, temporal
memory capacity, and multi-frequency oscillator prediction, performance degrades overall with
increasing $\rho$ (Fig.~\ref{fig:rho_sweep}). On three
classic delayed-memory tasks, namely the Adding Problem, delayed match-to-sample,
and pattern hold, all seven reported task--delay conditions show
significant trends in the opposite direction (Fig.~\ref{fig:tradeoff}), although
the curves are not pointwise monotonic and long-delay DMS remains near chance.
Perturbation recovery is consistent with the memory-maintenance benefit:
networks with higher reciprocity suppress deviations from a reference
trajectory more rapidly in the maximum-real operating regime
(Figs.~\ref{fig:functional_triad}b and \ref{fig:mechanism}b,d). It is not, by
itself, a memory assay, and the frozen-readout integration experiment shows
that faster geometric recovery can coexist with non-monotonic behavioural
	retention. The post-cue DMS scan reaches the same distinction under ongoing
	post-cue process noise: higher $\rho$ suppresses paired state deviations, but cue-margin
compression offsets that benefit at the decoder. Theorem
~\ref{thm:s5_attractor_memory} makes the additional
requirement explicit: within-cue deviations must contract without erasing the
separation between different cue-conditioned states or trajectories.

	The fixed edge budget provides a structural interpretation of how reciprocity
	can coordinate these three properties. Equation~\ref{eq:partner_budget} shows
	that, in the present ensemble, increasing $\rho$ reallocates directed edges
	from distinct-partner coverage $U$ to feedback closure. Broader partner coverage
	can diversify propagation paths and driven responses, whereas feedback loops
	can strengthen local restoration. The same convergence can compress
	cue-conditioned trajectories, making task-state margin a separate limiting
	factor. This resource-reallocation account is coupled to the dynamical
	operating point: changing spectral normalization removes the observed
	restoration trend without removing the capacity trend. It is therefore a
	structural interpretation of the present finite-budget ensemble, not an
	identified causal decomposition of reciprocity and partner coverage.

These results motivate a cautious interpretation of biological connectivity. The
seven \newrev{connectome and microcircuit datasets} examined occupy an intermediate reciprocity regime
(Fig.~\ref{fig:species}), rather than either endpoint of the synthetic sweep.
This positioning is compatible with, but does not prove,
balancing the opposing pressures identified in the model. The adult
\textit{Drosophila} region analysis shows substantial anatomical heterogeneity
(Fig.~\ref{fig:species}b); several putative feedforward groups are strongly
unidirectional, but the deposited region classes do not form a clean functional
ordering and differ greatly in size and sampling.

The over-representation of reciprocal motifs relative to ER null models
($8\times$ to $>10^3\times$), despite intermediate \emph{absolute} reciprocity,
has a structural explanation: in sparse networks, the ER baseline expectation
of reciprocity is extremely low, and spatial embedding with heterogeneous
connection probabilities mathematically inflates motif counts relative to
uniform random graphs~\citep{hoffmann2017nonrandom}. Reciprocal connections have
also been implicated in complementary functions, including persistent activity~\citep{fuster1971neuron,
goldman1995cellular}, pattern completion~\citep{sommer2001associative}, and
attractor stabilisation~\citep{hopfield1982neural}. These functions can matter when
circuits must maintain information without continuous input. Our results
suggest that these functions can be in tension with both the representational
capacity and task-state separation supported by more directed connectivity.
Reciprocity is therefore one structural coordinate whose consequences depend
on finite wiring resources, task demands, and operating regime.

\newrev{Recent work makes the dependence on circuit organization particularly
clear. \citet{ghosh2026partial} varied pathways in trained architectures and
found task-dependent benefits of feedback; \citet{dahmen2026localized} linked
cortical dimensionality to effective recurrence and population-specific motifs;
and \citet{javadzadeh2026consensus} identified selective persistence of agreement
modes in reciprocal inter-area circuits. These studies concern pathway
architecture, effective coupling, or inter-area mode structure, respectively.
Our contribution is to examine reciprocal-edge allocation at a fixed edge
budget and to relate activity geometry and driven-task performance to stability
and cue separation, including paired measurements of online MC, stability, and
cue margin on the same recurrent matrices.}

\newrev{This account also connects to established theory of response geometry,
memory, and perturbation amplification.} \citet{wainrib2015context} found
that random symmetric connectivity retained less contextual information than
asymmetric connectivity in a linear model, but that this ordering could reverse
in the nonlinear regime. That conditional result is consistent with our formal
finding that input alignment and operating regime are indispensable for any
SSV ordering. Likewise, \citet{ganguli2008memory} showed that appropriately
aligned non-normal architectures can provide extensive online trace memory;
this supports our placement of continuously driven MC on the transient-encoding
side and cautions against equating all memory definitions. For stable linear systems, the Kreiss Matrix
Theorem~\citep{trefethen2005spectra} establishes the complementary possibility
of large transient perturbation amplification, and our recovery experiment
tests this prediction in the nonlinear operating regime.
\citet{shao2023relating,shao2025impact}
demonstrated that reciprocal and chain motifs map onto specific spectral
signatures. Our paired-edge lemma supplies an exact topology-to-antisymmetry
bridge before E/I balancing, while the measured
relationship between $\rho$ and $\delta_{\mathrm{rel}}(\mathbf{W})$ supplies an empirical
bridge for the final signed ensemble. Together, these theoretical threads
distinguish three objects: response geometry, perturbation amplification, and
cue-conditioned decoding margin. Their joint effect motivates an empirically
testable trade-off whose expression depends on input alignment, task, and
operating regime. Our empirical contribution is to map those dependencies
across synthetic tasks; the connectome analysis supplies
structural context rather than an independent functional validation, and the
commutator defect remains a diagnostic rather than an identified mediator.

The dimensionality result is robust across a broad range of controls reported
in the Supplementary Materials: the relationship between $\rho$ and SSV holds across network
sizes ($N = 100$--$500$; Supplementary Note~S1), topologies and connection
densities (S3), and is
not attributable to proximity to the edge of chaos (S2). Reciprocity
compresses the effective rank of the recurrent weight matrix without
reducing its algebraic rank (S4). The trend between $\rho$ and SSV also survives global,
row-balanced, and Dale-like E/I implementations (S6), although the Dale-like
control shows that $\delta_{\mathrm{rel}}$ is not a universal or unique
mediator of the dimensionality effect. This robustness statement does not
extend to the stability result, whose $\rho$ trend is absent under the paired
maximum-modulus operating point (S7).

Several limitations should be noted. Our networks employ fixed recurrent
weights, whereas biological plasticity can dynamically reconfigure reciprocity
on behavioural timescales. The reciprocity parameter $\rho$ is a global scalar,
while biological reciprocity is structured across cortical layers, cell types,
and projection classes. In the fixed-edge-count generator, $\rho$ is also
coupled algebraically to the number of distinct connected partner pairs.
Fixed-skeleton or degree-preserving interventions are required to separate
partner coverage from reciprocity as causal variables. The ER null used for motif over-representation is a
descriptive baseline; stronger claims about motif enrichment should be tested
against degree-preserving and spatially constrained null models. No genuine
continuous-time rate-model or spiking-network control is reported here;
whether the relationship between $\rho$ and SSV generalizes to those dynamics remains an
open test. Our connectome analyses also rely on structural adjacency matrices;
effective connectivity, shaped by plasticity,
neuromodulation, and inhibitory gating, may differ substantially from the
structural substrate analysed here.

The formal theory is exact for finite-horizon
linear responses and for stable stationary linear systems, but it does not
assert a universal monotone map from reciprocity or non-normality to SSV or
memory; input alignment, noise structure, and the nonlinear operating point
remain essential (Supplementary Note~S5). The current recovery experiment
characterises one bias-driven reference trajectory and does not measure
cue-conditioned separation, basin retention, or transverse contraction during
DMS and Pattern Hold. Cue margin was measured directly in the post-cue DMS
scan, but not across every delayed-maintenance task. In addition, the weak or absent within-$\rho$
associations show that the present commutator analysis is an ensemble-level
mechanistic correlate, not a causal mediation test. \rev{The component-matching
analysis is also conditional on its implementation: matching $D$ and $V$ by
changing cue and process-noise amplitudes can alter the task operating point and
other unmeasured aspects of state geometry. It therefore tests the proposed
components under a specified protocol rather than providing a general
intervention on reciprocity.}

The shared-matrix
confirmation further shows that the stability trend is absent under the
paired maximum-modulus normalization and that neither input-noise robustness nor
perturbed-task retention increases monotonically with $\rho$. These null and
regime-dependent results limit the structural-control-parameter interpretation
to the measured three-way relationship rather than a universal robustness law.
The post-cue DMS scan establishes a robust reduction in internal
noise amplification, but its decoding curves likewise lack a $\rho$-endpoint
advantage because cue separation changes concurrently. \reviewrev{Under the controlled
matching protocol, higher reciprocity improved frozen-readout robustness after
cue margin was matched, and matching within-cue dispersion attenuated that gain.
These results identify a stability-associated contribution under the tested
conditions and support the proposed component-wise
interpretation of behavioural robustness.}
	
\section{Conclusion}
\label{sec:conclusion}

Our results show that reciprocity $\rho$ coordinates a three-way,
task- and operating-regime-dependent dynamical trade-off among representational
capacity, local stability, and task-state separability. Low reciprocity supports
larger representational volume and stronger high-dimensional transient
computation. Higher reciprocity accelerates local recovery and reduces internal
noise propagation in the original maximum-real operating regime, but also
compresses cue-state margins. The paired maximum-modulus control shows that the
local-stability component depends on operating point, while the post-cue scan
shows why greater contraction need not produce a monotonic gain in task
accuracy: behavioural robustness depends on whether reduced noise amplification
compensates for the loss of state separation. \reviewrev{Component matching further
yields a reproducible stability-associated contrast in frozen-readout
robustness.} \newrev{The shared-matrix trajectory uses online MC as its
task-specific $C$ coordinate, while SSV characterizes driven response geometry.}
Across
$\rho$-group means, the
\newrev{SSV} and recovery trends covary with the relative commutator normality
defect $\delta_{\mathrm{rel}}(\mathbf{W})$ in the generated ensemble, without
evidence that this quantity is a sufficient causal mediator. Under the fixed
edge budget used here, reciprocity also reallocates connections from
distinct-partner coverage to feedback closure; whether biological circuits use
this allocation remains a testable prediction. \newrev{Biological connectivity datasets} occupy
an intermediate-reciprocity regime, consistent with but not by itself proving
a balance among these pressures. Reciprocity is therefore neither uniformly
beneficial nor uniformly detrimental: its behavioural consequences depend
jointly on response diversity, local contraction, task-state separation, and
the dynamical operating regime.




\bibliography{references}
	\clearpage
	
\section*{Supplementary Materials}
\addcontentsline{toc}{section}{Supplementary Materials}

\subsection*{S1. Robustness across network sizes and topologies}

Unless explicitly stated otherwise, the legacy controls in S1--S4 use the
same maximum-real normalization as the main task sweep. They test robustness
to size, topology, density, and metric choice within that operating regime;
they do not establish robustness to conventional maximum-modulus normalization,
which is tested separately in S6--S7.

To test whether the relationship between $\rho$ and SSV depends on network scale, we
repeated the $\rho$ sweep at four network sizes ($N \in \{100, 200, 300,
500\}$; $c = 0.1$, maximum-real target $\sigma = 1.0$, $n = 6$ trials per condition). The
negative correlation between $\rho$ and SSV remained strong across all $N$
(Table~\ref{tab:n_rho}). Task performance similarly retained its dependence
on $\rho$ at every network size, with SAD accuracy declining and MG NMSE
rising as reciprocity increased.

\begin{table}[ht]
	\centering
	\caption{\textbf{Reciprocity effect on SSV and task performance across
			network sizes.} Pearson $r$ between $\rho$ and each metric. All
		correlations significant at $p < 0.001$.}
	\label{tab:n_rho}
	\begin{tabular}{c c c c}
		\toprule
		$N$ & SSV vs.\ $\rho$ ($r$) & SAD vs.\ $\rho$ ($r$) & MG vs.\ $\rho$ ($r$) \\
		\midrule
		100 & $-$0.994 & $-$0.978 & $+$0.955 \\
		200 & $-$0.988 & $-$0.953 & $+$0.938 \\
		300 & $-$0.991 & $-$0.993 & $+$0.948 \\
		500 & $-$0.997 & $-$0.972 & $+$0.877 \\
		\bottomrule
	\end{tabular}
\end{table}



\subsection*{S2. Trial-to-trial variability}

To assess whether the effects attributed to reciprocity could instead be
summarized by drift in maximum eigenvalue modulus, we examined
trial-to-trial variability and the relationship between $|\lambda|_{\max}$
and task performance. The maximum absolute eigenvalue $|\lambda|_{\max}$
correlated weakly with MG NMSE ($r = -0.41$; Supplementary Fig.~S1b),
making a simple one-variable modulus explanation less likely within this
ensemble. This observational analysis does not replace the paired
maximum-modulus control in S7. Trial-to-trial variability was
largest at low reciprocity and declined with increasing $\rho$
(Supplementary Fig.~S1a), consistent with the prediction that non-normal
dynamics amplify noise.

\begin{figure}[htbp]
	\centering
	\includegraphics[width=\textwidth]{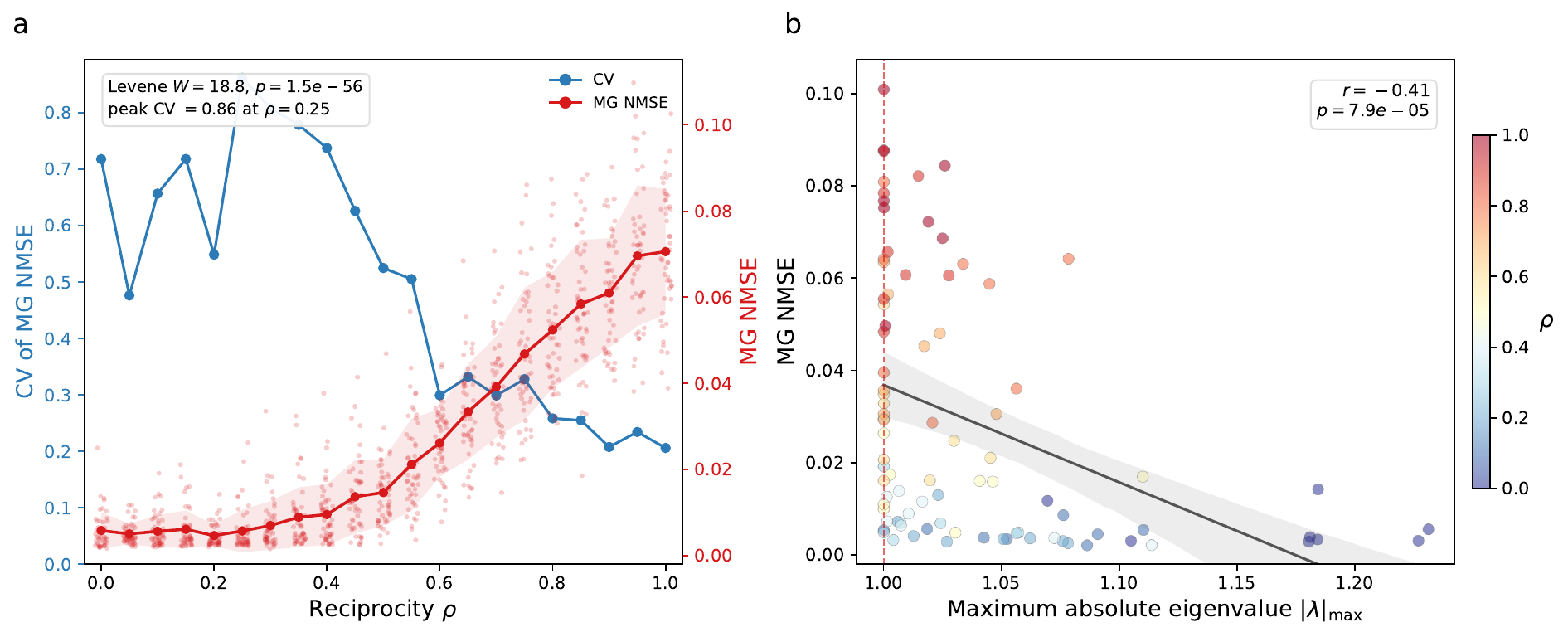}
	\caption{\textbf{Trial-to-trial variability.}
		\textbf{(a)} Coefficient of variation of MG NMSE and per-trial NMSE vs.\
		$\rho$. \textbf{(b)} $|\lambda|_{\max}$ vs.\ MG NMSE, coloured by $\rho$.
		$N = 300$, $c = 0.1$, maximum-real target $\sigma = 1.0$,
		$n = 50$ per $\rho$.}
	\label{fig:supp_variability}
\end{figure}

\subsection*{S3. Generality across topologies and densities}

To verify that the relationship between $\rho$ and performance is not an artefact of a
specific network construction algorithm, we tested three topologies (ER,
cycle-controlled, reciprocity-controlled) at six connection densities
($c \in \{0.001, \dots, 0.2\}$). Across all topologies and densities,
increasing $\rho$ consistently reduced SAD accuracy and increased MG NMSE
(Supplementary Fig.~S2). The effect was strongest in the
reciprocity-controlled topology and remained qualitatively present in all
alternatives. This supports robustness across the tested constructions within
the maximum-real operating regime rather than a normalization-independent
property of recurrent connectivity.

\begin{figure}[htbp]
	\centering
	\includegraphics[width=\textwidth]{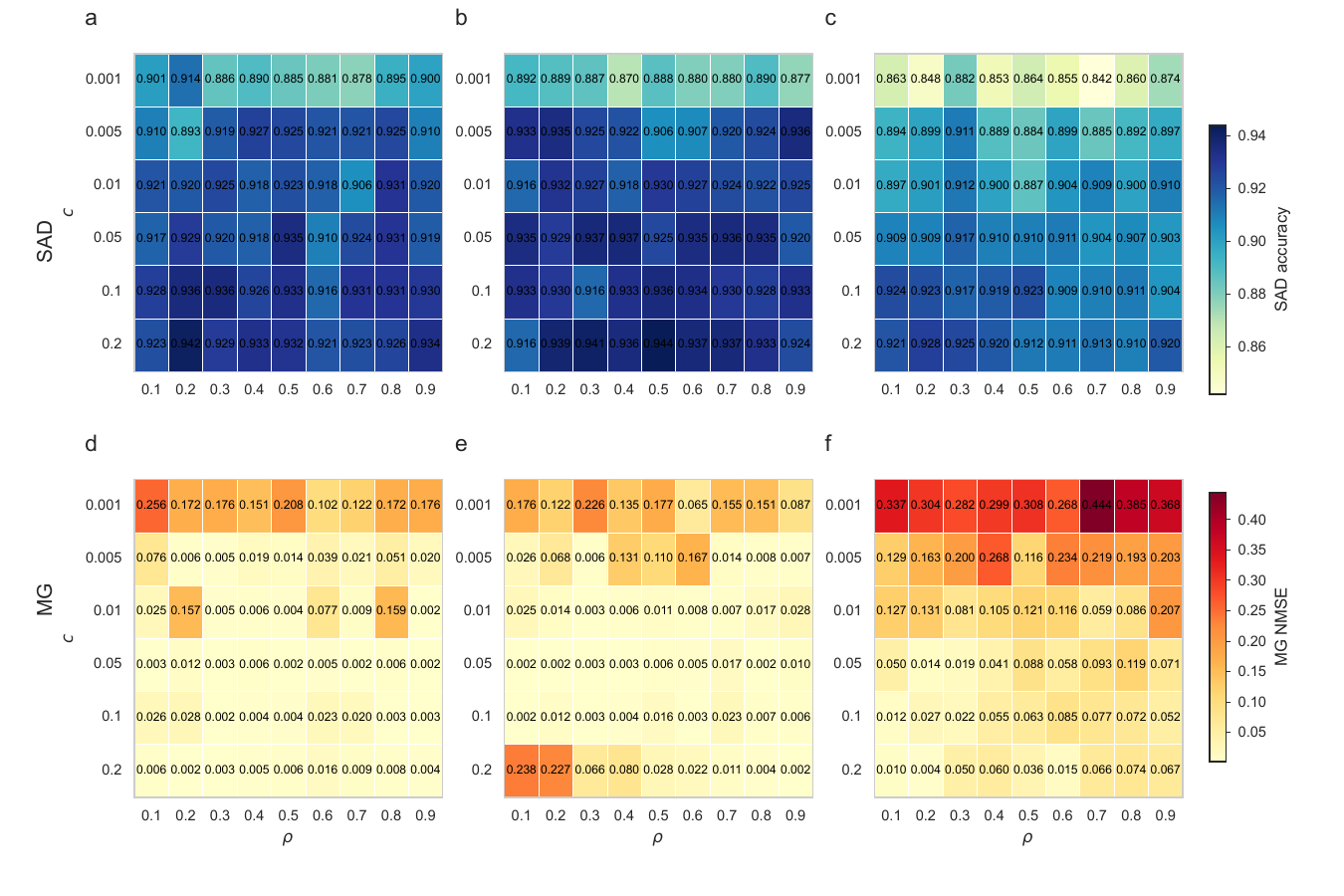}
	\caption{\textbf{Topology and density dependence.}
		\textbf{(a--c)} SAD accuracy for ER, Cycles, and Reciprocity topologies.
		\textbf{(d--f)} MG NMSE. $N = 300$, maximum-real target
		$\sigma = 1.0$.}
	\label{fig:supp_topology}
\end{figure}

\subsection*{S4. Spectral mechanisms and effective rank}

To link the relationship between $\rho$ and SSV to spectral structure, we computed the
effective rank of $\mathbf{W}$ via the participation ratio
$\mathrm{PR} = (\sum_i |\lambda_i|)^2 / \sum_i |\lambda_i|^2$.
Increasing $\rho$ reduced PR from $264 \pm 1$ to $224 \pm 1$
($r = -0.98$, $p < 0.001$; Supplementary Fig.~S3), while the matrix rank
remained identically $N = 300$, confirming that reciprocity compresses the
effective dimensionality of the recurrent weight matrix without reducing
its algebraic rank. An independent spectral proxy $H$ tracked SSV
($r = 0.84$), and weight reciprocity ($\alpha$) independently degraded MG
NMSE at fixed link reciprocity ($r = 0.85$, $p < 0.001$), confirming that
both link-level and weight-level reciprocity contribute to the trade-off.

\begin{figure}[htbp]
		\centering
		\includegraphics[width=\textwidth]{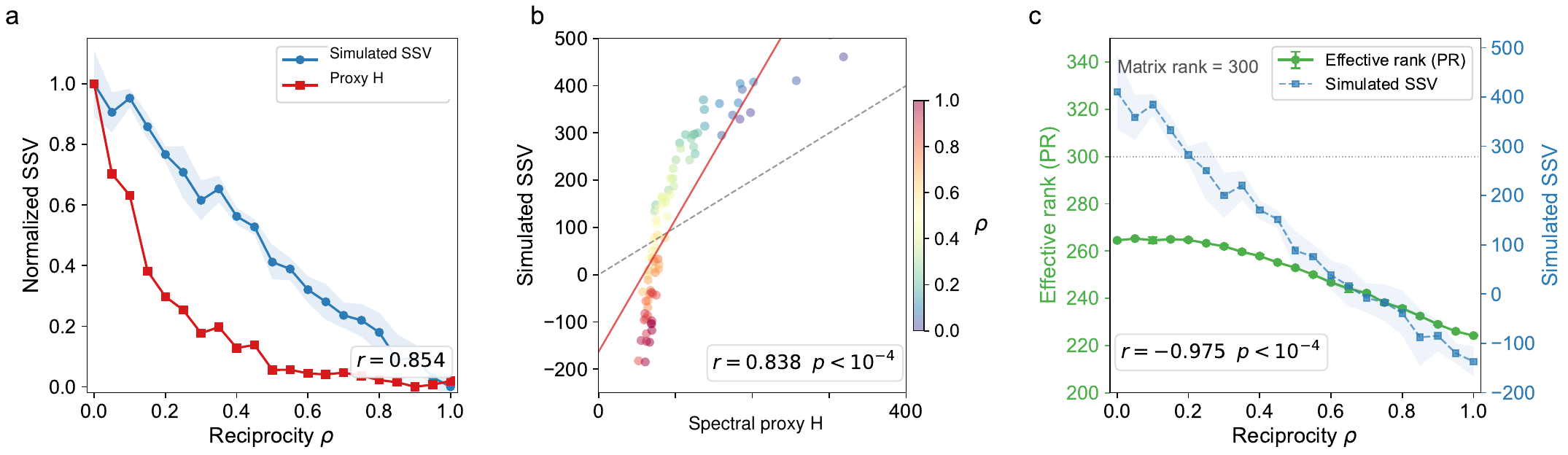}
	\caption{\textbf{Spectral proxies and effective rank.}
		\textbf{(a)} SSV and spectral proxy $H$ vs.\ $\rho$.
		\textbf{(b)} $H$ vs.\ measured SSV.
		\textbf{(c)} Effective rank (PR) and SSV vs.\ $\rho$.}
		\label{fig:supp_spectral}
	\end{figure}

\subsection*{S5. Formal theory for standardized state-space volume}
\label{sec:supp_formal_theory}

This note states precisely what can and cannot be concluded from linear
network theory. The SSV implemented in the code is computed after separately
standardizing every neuronal trace. It is therefore a log-determinant of a
\emph{correlation} matrix, not of the raw covariance Gramian. Keeping this
normalization explicit is essential: non-normality can enlarge raw covariance
while either enlarging or reducing standardized SSV, depending on how the
input is aligned with the recurrent dynamics. In the centered linear analysis
below, the constant input bias is absorbed into the state mean; $\mathbf{B}$
maps only the zero-mean fluctuating input.
\newrev{The results in this note concern response-volume geometry, noise
amplification, and delayed-cue decoding under their stated assumptions. They do
not identify standardized SSV with the summed online memory score $C$ in
Fig.~\ref{fig:functional_triad}, which additionally depends on the input protocol,
tested delays, and trained readout.}

\begin{definition}[Linear response Gramians and standardized SSV]
\label{def:s5_standardized_ssv}
Let
\begin{equation}
    \mathbf{x}_{t+1}=\mathbf{A}\mathbf{x}_{t}+\mathbf{B}\mathbf{u}_{t},
    \qquad
    \mathbb{E}\mathbf{u}_{t}=0,
    \qquad
    \mathbb{E}\mathbf{u}_{t}\mathbf{u}_{s}^{\top}
    =\delta_{ts}\mathbf{I}_{m},
    \label{eq:s5_linear_system}
\end{equation}
where $\mathbf{A}\in\mathbb{R}^{N\times N}$,
$\mathbf{B}\in\mathbb{R}^{N\times m}$, and
$\mathbf{Q}=\mathbf{B}\mathbf{B}^{\top}$. For $L\geq0$, define
\begin{equation}
    \mathbf{G}_{L}
    =\sum_{k=0}^{L}\mathbf{A}^{k}\mathbf{Q}(\mathbf{A}^{\top})^{k},
    \qquad
    \mathbf{D}_{L}=\operatorname{diag}
       ((G_L)_{11},\ldots,(G_L)_{NN}).
    \label{eq:s5_finite_gramian}
\end{equation}
Whenever every diagonal entry is positive, let
$\mathbf{R}_{L}=\mathbf{D}_{L}^{-1/2}\mathbf{G}_{L}
\mathbf{D}_{L}^{-1/2}$. If $\mathbf{R}_{L}\succ0$, define the
population analogue of the code's finite-sample statistic
\begin{equation}
    \mathcal{V}_{\mathrm{std}}(\mathbf{R}_{L};T)
    =\frac{N}{2}\log(2\pi T)+\frac{1}{2}\log\det\mathbf{R}_{L}.
    \label{eq:s5_standardized_volume}
\end{equation}
The constant $T>0$ represents the fixed number of post-washout samples and
does not affect comparisons. If $\mathbf{R}_{L}$ is singular, the
unregularized value is $-\infty$; the numerical eigenvalue floor used in the
experiments gives its regularized finite-sample counterpart.
\end{definition}

\begin{theorem}[Covariance Gramian and exact population-volume identity]
\label{thm:standardized_gramian}
For Eq.~\ref{eq:s5_linear_system} with $\mathbf{x}_{0}=0$:
\begin{enumerate}[label=(\roman*)]
    \item $\operatorname{Cov}(\mathbf{x}_{L+1})=\mathbf{G}_{L}$ for every
    finite $L$, without a stability assumption. Moreover,
    \begin{equation}
        \mathbb{E}\!\left[\frac{1}{T}\sum_{t=1}^{T}
        \mathbf{x}_{t}\mathbf{x}_{t}^{\top}\right]
        =\sum_{k=0}^{T-1}\left(1-\frac{k}{T}\right)
        \mathbf{A}^{k}\mathbf{Q}(\mathbf{A}^{\top})^{k}.
        \label{eq:s5_average_gramian}
    \end{equation}
    \item If the spectral radius $\varrho(\mathbf{A})<1$, then
    $\mathbf{G}_{\infty}=\sum_{k=0}^{\infty}
    \mathbf{A}^{k}\mathbf{Q}(\mathbf{A}^{\top})^{k}$ converges and is the
    unique positive-semidefinite solution of
    \begin{equation}
        \mathbf{G}_{\infty}
        =\mathbf{A}\mathbf{G}_{\infty}\mathbf{A}^{\top}+\mathbf{Q}.
        \label{eq:s5_lyapunov}
    \end{equation}
    It is positive definite if and only if
    $(\mathbf{A},\mathbf{B})$ is controllable.
    \item For any positive-definite Gramian $\mathbf{G}$ with
    $\mathbf{D}=\operatorname{diag}(G_{11},\ldots,G_{NN})$,
    \begin{equation}
        \log\det(\mathbf{D}^{-1/2}\mathbf{G}\mathbf{D}^{-1/2})
        =\log\det\mathbf{G}-\sum_{i=1}^{N}\log G_{ii}\leq0.
        \label{eq:s5_logdet_identity}
    \end{equation}
    Equality holds if and only if $\mathbf{G}$ is diagonal.
\end{enumerate}
\end{theorem}

\begin{proof}
Iterating Eq.~\ref{eq:s5_linear_system} gives
$\mathbf{x}_{L+1}=\sum_{k=0}^{L}\mathbf{A}^{k}
\mathbf{B}\mathbf{u}_{L-k}$. Temporal whiteness eliminates all cross terms,
yielding $\operatorname{Cov}(\mathbf{x}_{L+1})=\mathbf{G}_{L}$. Summing this
identity over $t=1,\ldots,T$ shows that the term with lag $k$ appears $T-k$
times, which proves Eq.~\ref{eq:s5_average_gramian}.

If $\varrho(\mathbf{A})<1$, there are $C<\infty$ and $0<\beta<1$ such that
$\|\mathbf{A}^{k}\|_{2}\leq C\beta^{k}$. The Gramian series therefore
converges absolutely. Shifting its index proves Eq.~\ref{eq:s5_lyapunov}. If
$\boldsymbol{\Delta}$ is the difference of two solutions, then
$\boldsymbol{\Delta}=\mathbf{A}^{k}\boldsymbol{\Delta}
(\mathbf{A}^{\top})^{k}$ for every $k$, and taking $k\to\infty$ gives
$\boldsymbol{\Delta}=0$. Finally,
\begin{equation}
    \mathbf{v}^{\top}\mathbf{G}_{\infty}\mathbf{v}
    =\sum_{k=0}^{\infty}
    \|\mathbf{B}^{\top}(\mathbf{A}^{\top})^{k}\mathbf{v}\|_{2}^{2},
\end{equation}
which is positive for every nonzero $\mathbf{v}$ exactly when the
controllability matrix has row rank $N$~\citep{kailath1980linear}.

For part (iii), determinant multiplicativity gives the equality in
Eq.~\ref{eq:s5_logdet_identity}. Hadamard's determinant inequality gives
$\det\mathbf{G}\leq\prod_iG_{ii}$, with equality for a positive-definite
matrix exactly when its coordinate vectors are orthogonal, equivalently when
$\mathbf{G}$ is diagonal~\citep{horn2012matrix}.
\end{proof}

Theorem~\ref{thm:standardized_gramian} shows why the z-scoring step matters.
Raw amplification that multiplies one or several neuronal variances is removed
by the term $\sum_i\log G_{ii}$. The remaining determinant measures whether
the temporal responses of different neurons are mutually distinct. Note that
Eq.~\ref{eq:s5_average_gramian} is the expectation of an uncentered finite-time
second moment; normalization and $\log\det$ do not commute with expectation, so
it is not an exact identity for the random finite-sample SSV. The connection to
the sample estimator is asymptotic under the assumptions below.

\begin{corollary}[Consistency of the sample SSV in the stable linear regime]
\label{cor:s5_sample_consistency}
Assume $\varrho(\mathbf{A})<1$, the inputs in
Eq.~\ref{eq:s5_linear_system} are i.i.d. with finite second moment, and
$(\mathbf{A},\mathbf{B})$ is controllable. For any almost-surely finite initial
state, let $\widehat{\mathbf{G}}_{T}$ be the centered sample covariance of
$\mathbf{x}_{1},\ldots,\mathbf{x}_{T}$ and let
$\widehat{\mathbf{R}}_{T}$ be its correlation matrix. Then, almost surely,
\begin{equation}
    \widehat{\mathbf{G}}_{T}\longrightarrow\mathbf{G}_{\infty},
    \qquad
    \widehat{\mathbf{R}}_{T}\longrightarrow\mathbf{R}_{\infty},
    \qquad
    \SSV_T-\frac{N}{2}\log(2\pi T)
    \longrightarrow\frac{1}{2}\log\det\mathbf{R}_{\infty}.
    \label{eq:s5_sample_consistency}
\end{equation}
\end{corollary}

\begin{proof}
Under $\varrho(\mathbf{A})<1$, the stationary solution is the causal linear
process
$\mathbf{x}_{t}=\sum_{k=0}^{\infty}\mathbf{A}^{k}
\mathbf{B}\mathbf{u}_{t-1-k}$. It is a measurable shift-equivariant function
of an i.i.d. sequence and is therefore stationary and ergodic. The ergodic
theorem, applied entrywise to $\mathbf{x}_{t}$ and
$\mathbf{x}_{t}\mathbf{x}_{t}^{\top}$, gives convergence of the sample mean
and covariance; the transient from any finite initial state vanishes
geometrically. Controllability gives $\mathbf{G}_{\infty}\succ0$, so diagonal
normalization and $\log\det$ are continuous in a neighbourhood of the limit.
The continuous mapping theorem then gives the remaining limits.
\end{proof}

\begin{corollary}[Krylov-response geometry]
\label{cor:s5_krylov_geometry}
Let
\begin{equation}
    \mathbf{K}_{L}=
    [\,\mathbf{B},\mathbf{A}\mathbf{B},\ldots,\mathbf{A}^{L}\mathbf{B}\,]
    \in\mathbb{R}^{N\times m(L+1)}.
    \label{eq:s5_krylov_matrix}
\end{equation}
Assume every row of $\mathbf{K}_{L}$ is nonzero. Then
$\mathbf{G}_{L}=\mathbf{K}_{L}\mathbf{K}_{L}^{\top}$. If
$\widehat{\mathbf{K}}_{L}$ is obtained by normalizing every nonzero row of
$\mathbf{K}_{L}$ to unit Euclidean norm, then
\begin{equation}
    \mathbf{R}_{L}=\widehat{\mathbf{K}}_{L}
    \widehat{\mathbf{K}}_{L}^{\top},
    \qquad
    \det\mathbf{R}_{L}
    =\operatorname{vol}_{N}^{2}
      (\widehat{\mathbf{k}}_{1},\ldots,\widehat{\mathbf{k}}_{N})
    =\prod_{i=2}^{N}\sin^{2}\theta_i,
    \label{eq:s5_response_volume}
\end{equation}
where $\theta_i$ is the angle between row $i$ and the span of rows
$1,\ldots,i-1$ during Gram--Schmidt orthogonalization. Hence standardized SSV
is finite exactly when the normalized neuronal impulse-response signatures
span an $N$-dimensional row space.
\end{corollary}

\begin{proof}
The first identity follows by multiplying the block matrix in
Eq.~\ref{eq:s5_krylov_matrix} by its transpose. Row normalization is precisely
the congruence transformation by $\mathbf{D}_{L}^{-1/2}$, proving the second
identity. The determinant of a Gram matrix is the squared volume of its
generating vectors; Gram--Schmidt expresses that volume as the product of the
orthogonal residual lengths. Because every row has unit norm, those lengths
are $1,\sin\theta_2,\ldots,\sin\theta_N$.
\end{proof}

\begin{proposition}[Exact raw covariance for normal dynamics]
\label{prop:s5_normal_covariance}
Suppose $\mathbf{A}$ is real normal, $\varrho(\mathbf{A})<1$, and
$\mathbf{Q}=q\mathbf{I}$ with $q>0$. If $\mu_1,\ldots,\mu_N$ are the
eigenvalues of $\mathbf{A}$, then for a unitary diagonalization
$\mathbf{A}=\mathbf{U}\operatorname{diag}(\mu_i)\mathbf{U}^{*}$,
\begin{equation}
    \mathbf{G}_{\infty}
    =\mathbf{U}\operatorname{diag}\!\left(
      \frac{q}{1-|\mu_1|^2},\ldots,
      \frac{q}{1-|\mu_N|^2}\right)\mathbf{U}^{*},
    \qquad
    \log\det\mathbf{G}_{\infty}
    =N\log q-\sum_i\log(1-|\mu_i|^2).
    \label{eq:s5_normal_formula}
\end{equation}
\end{proposition}

\begin{proof}
Insert the unitary diagonalization into the convergent Gramian series and sum
the scalar geometric series in each eigendirection.
\end{proof}

\begin{theorem}[Normal dynamics attain the fixed-spectrum isotropic-noise floor]
\label{thm:s5_normal_noise_floor}
Let $\mathbf{A}\in\mathbb{R}^{N\times N}$ have eigenvalues
$\lambda_1,\ldots,\lambda_N$, counted with algebraic multiplicity, and
consider the error dynamics
\begin{equation}
    \mathbf{e}_{t+1}=\mathbf{A}\mathbf{e}_{t}
    +\boldsymbol{\xi}_{t},
    \qquad \mathbf{e}_0=\mathbf{0},
    \qquad
    \mathbb{E}\boldsymbol{\xi}_{t}=\mathbf{0},
    \qquad
    \mathbb{E}\boldsymbol{\xi}_{t}\boldsymbol{\xi}_{s}^{\top}
    =\sigma_{\xi}^{2}\delta_{ts}\mathbf{I}_{N}.
    \label{eq:s5_isotropic_error_system}
\end{equation}
For every integer $L\geq1$,
\begin{align}
    \mathbb{E}\|\mathbf{e}_{L}\|_2^2
    &=\sigma_{\xi}^{2}\sum_{k=0}^{L-1}
      \|\mathbf{A}^{k}\|_{F}^{2}                                      
      \label{eq:s5_noise_energy_identity}\\
    &\geq\sigma_{\xi}^{2}\sum_{i=1}^{N}
      \sum_{k=0}^{L-1}|\lambda_i|^{2k}.
      \label{eq:s5_noise_energy_floor}
\end{align}
If $L\geq2$, equality holds if and only if $\mathbf{A}$ is normal. Hence,
among all matrices with the same spectrum, every normal representative
minimizes the expected terminal energy generated by isotropic process noise.
If in addition $\varrho(\mathbf{A})<1$, then
\begin{equation}
    \lim_{L\to\infty}\mathbb{E}\|\mathbf{e}_{L}\|_2^2
    \geq \sigma_{\xi}^{2}\sum_{i=1}^{N}
    \frac{1}{1-|\lambda_i|^2},
    \label{eq:s5_stationary_noise_floor}
\end{equation}
again with equality if and only if $\mathbf{A}$ is normal.
\end{theorem}

\begin{proof}
Iteration of Eq.~\ref{eq:s5_isotropic_error_system} and independence give
\begin{equation}
    \operatorname{Cov}(\mathbf{e}_{L})
    =\sigma_{\xi}^{2}\sum_{k=0}^{L-1}
      \mathbf{A}^{k}(\mathbf{A}^{\top})^{k}.
\end{equation}
Taking the trace proves Eq.~\ref{eq:s5_noise_energy_identity}. Let
$\mathbf{A}=\mathbf{U}\mathbf{T}\mathbf{U}^{*}$ be a complex Schur
decomposition, where $\mathbf{U}$ is unitary, $\mathbf{T}$ is upper
triangular, and $T_{ii}=\lambda_i$. Unitary invariance of the Frobenius norm
and the fact that the diagonal of $\mathbf{T}^{k}$ is
$(\lambda_1^k,\ldots,\lambda_N^k)$ imply
\begin{equation}
    \|\mathbf{A}^{k}\|_F^2
    =\|\mathbf{T}^{k}\|_F^2
    \geq\sum_{i=1}^{N}|\lambda_i|^{2k}.
\end{equation}
Summation proves Eq.~\ref{eq:s5_noise_energy_floor}. For $L\geq2$, equality
of the sum requires equality for $k=1$. But
\begin{equation}
    \|\mathbf{T}\|_F^2
    =\sum_i|\lambda_i|^2+\sum_{i<j}|T_{ij}|^2,
\end{equation}
so equality holds exactly when $\mathbf{T}$ is diagonal, equivalently when
$\mathbf{A}$ is unitarily diagonalizable and therefore normal
~\citep{horn2012matrix}. Under stability, monotone convergence permits
passage to the infinite sum, and summing the scalar geometric series gives
Eq.~\ref{eq:s5_stationary_noise_floor}.
\end{proof}

\begin{corollary}[Noise covariance and cue margin jointly determine decoding reliability]
\label{cor:s5_gaussian_decoding}
Suppose an equiprobable binary cue $z\in\{-1,+1\}$ produces
\begin{equation}
    \mathbf{y}_{L}\mid z
    \sim\mathcal{N}(z\mathbf{s}_{L},\boldsymbol{\Sigma}_{L}),
    \qquad \boldsymbol{\Sigma}_{L}\succ0.
\end{equation}
Then the optimal decoder has error probability
\begin{equation}
    P_{\mathrm{err}}
    =\Phi\!\left(-\sqrt{\mathbf{s}_{L}^{\top}
      \boldsymbol{\Sigma}_{L}^{-1}\mathbf{s}_{L}}\right),
    \label{eq:s5_binary_bayes_error}
\end{equation}
where $\Phi$ is the standard-normal distribution function. More generally,
for $C$ equiprobable Gaussian classes with means $\boldsymbol{\mu}_{c}$ and a
common covariance $\boldsymbol{\Sigma}_{L}$, define
\begin{equation}
    d_{\min}^{2}
    =\min_{c\neq d}(\boldsymbol{\mu}_{c}-\boldsymbol{\mu}_{d})^{\top}
      \boldsymbol{\Sigma}_{L}^{-1}
      (\boldsymbol{\mu}_{c}-\boldsymbol{\mu}_{d}).
    \label{eq:s5_mahalanobis_margin}
\end{equation}
The common-covariance Bayes decoder obeys
\begin{equation}
    P_{\mathrm{err}\mid c}
    \leq (C-1)\Phi(-d_{\min}/2).
    \label{eq:s5_multiclass_error_bound}
\end{equation}
\end{corollary}

\begin{proof}
For the binary model, the log-likelihood ratio has the sign of
$\mathbf{s}_{L}^{\top}\boldsymbol{\Sigma}_{L}^{-1}\mathbf{y}_{L}$.
Conditioned on $z=+1$, this scalar is Gaussian with both mean and variance
equal to
$\mathbf{s}_{L}^{\top}\boldsymbol{\Sigma}_{L}^{-1}\mathbf{s}_{L}$,
which proves Eq.~\ref{eq:s5_binary_bayes_error}; symmetry handles $z=-1$.
For multiple classes, the probability that class $d$ defeats the true class
$c$ at their pairwise Mahalanobis decision boundary is
$\Phi(-d_{cd}/2)$. The union bound over $d\neq c$ and
$d_{cd}\geq d_{\min}$ prove Eq.~\ref{eq:s5_multiclass_error_bound}.
\end{proof}

Equation~\ref{eq:s5_normal_formula} is the eigenvalue-only covariance formula
used to motivate the spectral proxy $H$ and is consistent with linear random
network covariance theory~\citep{hu2022spectrum}. It orders the \emph{raw}
covariance volume, not automatically the reported standardized SSV: the latter
also subtracts $\sum_i\log(G_{\infty})_{ii}$, which can depend on
$\mathbf{U}$. Eigenvalues alone order standardized SSV only under additional
conditions that also control these marginal variances.

\begin{theorem}[A sufficient spectral ordering for standardized SSV]
\label{thm:ssv_majorization}
Let $\mathbf{R}_{1},\mathbf{R}_{2}\succ0$ be two $N\times N$ correlation
matrices, with eigenvalues sorted in nonincreasing order. If
$\boldsymbol{\lambda}(\mathbf{R}_{2})$ majorizes
$\boldsymbol{\lambda}(\mathbf{R}_{1})$, that is,
\begin{equation}
    \sum_{i=1}^{k}\lambda_i(\mathbf{R}_{2})
    \geq\sum_{i=1}^{k}\lambda_i(\mathbf{R}_{1})
    \quad (k=1,\ldots,N-1),
    \qquad
    \sum_{i=1}^{N}\lambda_i(\mathbf{R}_{1})
    =\sum_{i=1}^{N}\lambda_i(\mathbf{R}_{2})=N,
    \label{eq:s5_majorization}
\end{equation}
then
\begin{equation}
    \mathcal{V}_{\mathrm{std}}(\mathbf{R}_{1};T)
    \geq
    \mathcal{V}_{\mathrm{std}}(\mathbf{R}_{2};T).
    \label{eq:s5_majorization_order}
\end{equation}
The inequality is strict unless the two eigenvalue vectors agree up to a
permutation.
\end{theorem}

\begin{proof}
The function $\log x$ is strictly concave on $(0,\infty)$. Karamata's
inequality applied to Eq.~\ref{eq:s5_majorization} therefore gives
$\sum_i\log\lambda_i(\mathbf{R}_{2})\leq
\sum_i\log\lambda_i(\mathbf{R}_{1})$, with strict inequality unless the
spectra coincide~\citep{horn2012matrix}. Adding the common constant in
Eq.~\ref{eq:s5_standardized_volume} proves the result.
\end{proof}

Thus an activity correlation spectrum that is more concentrated into a few
modes has rigorously smaller standardized SSV. Participation ratio is a useful
summary of this concentration, but equality or ordering of participation ratio
alone does not imply the full majorization condition.

\begin{theorem}[A fixed-spectrum trade-off between driven encoding and robust memory]
\label{thm:shear_standardized}
Let $0<r<1$, $h\geq0$, and
\begin{equation}
    \mathbf{A}_{h}=
    \begin{pmatrix}r&h\\0&r\end{pmatrix},
    \qquad
    \mathbf{b}=\begin{pmatrix}1\\1\end{pmatrix},
    \qquad
    \mathbf{x}_{t+1}=\mathbf{A}_{h}\mathbf{x}_{t}+\mathbf{b}u_t,
    \quad \mathbb{E}u_t=0,\quad
    \mathbb{E}[u_tu_s]=q\delta_{ts},\quad q>0.
    \label{eq:s5_shear_system}
\end{equation}
Every $\mathbf{A}_{h}$ has eigenvalues $(r,r)$ and spectral radius $r$.
For $h>0$, it is non-normal and the stationary correlation matrix satisfies
\begin{equation}
    \det\mathbf{R}_{h}
    =\frac{h^2}{(1-r^2+hr)^2+h^2}.
    \label{eq:s5_shear_correlation}
\end{equation}
This determinant is zero at $h=0$ and strictly increases for $h>0$ because
\begin{equation}
    \frac{d}{dh}\det\mathbf{R}_{h}
    =\frac{2h(1-r^2)(1-r^2+hr)}
    {\bigl[(1-r^2+hr)^2+h^2\bigr]^2}>0.
    \label{eq:s5_shear_derivative}
\end{equation}
Consequently, for every $0<h_1<h_2$,
$\mathcal{V}_{\mathrm{std}}(\mathbf{R}_{h_2};T)>
\mathcal{V}_{\mathrm{std}}(\mathbf{R}_{h_1};T)$; this strict finite-to-finite
ordering does not rely only on the rank change at $h=0$.
Moreover,
\begin{equation}
\Delta_{\mathrm{comm}}(\mathbf{A}_{h})=\sqrt{2}\,h^2,
\qquad
\delta_{\mathrm{rel}}(\mathbf{A}_{h})
=\frac{\sqrt{2}\,h^2}{2r^2+h^2},
\label{eq:s5_shear_defects}
\end{equation}
and both defects increase strictly with $h>0$. For $k\geq1$,
\begin{equation}
    \|\mathbf{A}_{h}^{k}\|_2
    \geq\sqrt{r^{2k}+k^2h^2r^{2k-2}},
    \label{eq:s5_shear_transient}
\end{equation}
so increasing the shear also increases an explicit lower bound on transient
gain.

Now consider a delayed-memory experiment on the same matrix. Let
$z\sim\mathcal{N}(0,\tau^2)$ be a scalar cue, set
$\mathbf{y}_{0}=z\mathbf{e}_{1}$, and evolve
\begin{equation}
    \mathbf{y}_{t+1}=\mathbf{A}_{h}\mathbf{y}_{t}
    +\boldsymbol{\xi}_{t},
    \qquad
    \boldsymbol{\xi}_{t}\overset{\mathrm{i.i.d.}}{\sim}
    \mathcal{N}(\mathbf{0},\sigma_{\xi}^{2}\mathbf{I}_{2}),
    \label{eq:s5_shear_memory_system}
\end{equation}
independently of $z$. For a delay $L\geq2$, define
\begin{equation}
    a_k=r^k,
    \qquad
    b_0=0,\quad b_k=kr^{k-1}\ (k\geq1),
    \qquad
    S_0=\sum_{k=0}^{L-1}a_k^2,\quad
    S_1=\sum_{k=0}^{L-1}a_kb_k,\quad
    S_2=\sum_{k=0}^{L-1}b_k^2.
    \label{eq:s5_shear_memory_sums}
\end{equation}
Then $D_L=S_0S_2-S_1^2>0$, and the Fisher information in
$\mathbf{y}_L$ about the cue value $z$ is
\begin{equation}
    \mathcal{J}_{L}(h)
    =\frac{r^{2L}S_0}
    {\sigma_{\xi}^{2}\left[S_0^2+h^2D_L\right]}.
    \label{eq:s5_shear_memory_fisher}
\end{equation}
Consequently,
\begin{equation}
    \frac{d}{dh}\mathcal{J}_{L}(h)
    =-\frac{2hr^{2L}S_0D_L}
    {\sigma_{\xi}^{2}\left[S_0^2+h^2D_L\right]^2}<0
    \qquad(h>0),
    \label{eq:s5_shear_memory_derivative}
\end{equation}
and the Gaussian mutual information
\begin{equation}
    \mathsf{I}(z;\mathbf{y}_L)
    =\frac{1}{2}\log\!\left(1+\tau^2\mathcal{J}_{L}(h)\right)
    \label{eq:s5_shear_memory_mutual_information}
\end{equation}
also decreases strictly with $h>0$. Thus the same fixed-spectrum
non-normal deformation can enlarge driven standardized response volume and
transient sensitivity while reducing the information robustly recoverable
after a noisy delay.
\end{theorem}

\begin{proof}
For $d=1-r^2$, the impulse responses are
$\mathbf{A}_{h}^{k}\mathbf{b}=(r^k+khr^{k-1},r^k)^{\top}$, with the second
term interpreted as zero at $k=0$. Using
\begin{equation}
    \sum_{k\geq0}r^{2k}=\frac{1}{d},\qquad
    \sum_{k\geq0}kr^{2k-1}=\frac{r}{d^2},\qquad
    \sum_{k\geq0}k^2r^{2k-2}=\frac{1+r^2}{d^3},
\end{equation}
the stationary Gramian is
\begin{equation}
    \mathbf{G}_{h}=q
    \begin{pmatrix}
    d^{-1}+2hrd^{-2}+h^2(1+r^2)d^{-3}
        &d^{-1}+hrd^{-2}\\
    d^{-1}+hrd^{-2}&d^{-1}
    \end{pmatrix}.
    \label{eq:s5_shear_gramian}
\end{equation}
Direct calculation gives $\det\mathbf{G}_{h}=q^2h^2/d^4$; dividing by the
product of its diagonal entries yields Eq.~\ref{eq:s5_shear_correlation}.
Differentiation gives Eq.~\ref{eq:s5_shear_derivative}. Finally,
\begin{equation}
    \mathbf{A}_{h}\mathbf{A}_{h}^{\top}
    -\mathbf{A}_{h}^{\top}\mathbf{A}_{h}
    =\operatorname{diag}(h^2,-h^2),
    \qquad
    \mathbf{A}_{h}^{k}=
    \begin{pmatrix}r^k&khr^{k-1}\\0&r^k\end{pmatrix}.
\end{equation}
The Frobenius norm of the displayed commutator is $\sqrt{2}h^2$, while
$\|\mathbf{A}_{h}\|_F^2=2r^2+h^2$, proving
Eq.~\ref{eq:s5_shear_defects}. Its relative term has derivative
$4\sqrt{2}hr^2/(2r^2+h^2)^2>0$ for $h>0$.
Taking the norm of $\mathbf{A}_{h}^{k}(0,1)^{\top}$ proves the final bound.

For the delayed-memory experiment,
$\mathbf{A}_{h}^{L}\mathbf{e}_{1}=r^{L}\mathbf{e}_{1}$ and iteration of
Eq.~\ref{eq:s5_shear_memory_system} gives
\begin{equation}
    \operatorname{Cov}(\mathbf{y}_{L}\mid z)
    =\sigma_{\xi}^{2}\sum_{k=0}^{L-1}
    \mathbf{A}_{h}^{k}(\mathbf{A}_{h}^{\top})^{k}
    =\sigma_{\xi}^{2}
    \begin{pmatrix}
        S_0+h^2S_2&hS_1\\
        hS_1&S_0
    \end{pmatrix}.
    \label{eq:s5_shear_memory_covariance}
\end{equation}
The two vectors $(a_0,\ldots,a_{L-1})$ and
$(b_0,\ldots,b_{L-1})$ are not collinear when $L\geq2$; strict
Cauchy--Schwarz therefore gives $D_L=S_0S_2-S_1^2>0$. Inverting the
$2\times2$ matrix in Eq.~\ref{eq:s5_shear_memory_covariance} and using
$\partial_z\mathbb{E}[\mathbf{y}_{L}\mid z]=r^L\mathbf{e}_1$ yields
Eq.~\ref{eq:s5_shear_memory_fisher}. Differentiation proves
Eq.~\ref{eq:s5_shear_memory_derivative}. Finally, a scalar Gaussian prior
observed through a linear Gaussian channel has mutual information
$\frac12\log(1+\tau^2\mathcal{J}_{L})$, proving
Eq.~\ref{eq:s5_shear_memory_mutual_information}.
\end{proof}

\begin{corollary}[Exact noisy binary-decoding ordering in the shear family]
\label{cor:s5_shear_bayes}
In the delayed-memory experiment of
Theorem~\ref{thm:shear_standardized}, replace the Gaussian cue by an
equiprobable binary cue $z\in\{-a,+a\}$ with $a>0$. The optimal terminal
decoder has
\begin{equation}
    P_{\mathrm{err}}(h)=\Phi\!\left(-a\sqrt{\mathcal{J}_{L}(h)}\right),
    \label{eq:s5_shear_bayes_error}
\end{equation}
and $P_{\mathrm{err}}(h)$ is strictly increasing for $h>0$. Thus, in this
fixed-spectrum family, reducing the non-normal shear strictly improves actual
classification reliability under post-cue process noise, not only Fisher or
mutual information.
\end{corollary}

\begin{proof}
Equation~\ref{eq:s5_binary_bayes_error} gives the displayed formula, and
Eq.~\ref{eq:s5_shear_memory_derivative} gives
$\mathcal{J}_{L}'(h)<0$. Since $\Phi$ is strictly increasing,
differentiating $-a\sqrt{\mathcal{J}_{L}(h)}$ proves the claim.
\end{proof}

\begin{remark}[Why this is not a universal non-normality ordering]
\label{rem:s5_no_universal_order}
Theorem~\ref{thm:shear_standardized} holds for a matched scalar drive whose
Krylov responses are made more independent by the shear. With isotropic drive
$\mathbf{Q}=q\mathbf{I}$ instead, $\mathbf{A}_{0}=r\mathbf{I}$ has
$\mathbf{R}_{0}=\mathbf{I}$ and therefore already attains the maximal possible
determinant $1$ from Hadamard's inequality; a shear cannot increase it and,
for $r,h>0$, introduces nonzero correlation and decreases it. Thus equal
spectral radius and greater non-normality do not imply greater SSV without an
assumption on input alignment or an explicit response-spectrum ordering.
\end{remark}

The two experiments in Theorem~\ref{thm:shear_standardized} deliberately use
different input protocols. This is not a technical nuisance: an ongoing drive
asks how many distinguishable response directions the reservoir generates,
whereas a delay task asks how well a previously written signal survives
subsequent process noise. The theorem is therefore a constructive trade-off
under explicit alignment and noise assumptions, not a universal ordering of
all memory definitions. In particular, non-normal delay lines can provide
extensive trace memory when the input is aligned with their feedforward
direction~\citep{ganguli2008memory}.

\begin{theorem}[Separated attracting memories convert contraction into a decoding guarantee]
\label{thm:s5_attractor_memory}
Let $\mathbf{F}:\mathbb{R}^{N}\to\mathbb{R}^{N}$ have cue-specific fixed
points $\mathbf{x}_{1}^{\star},\ldots,\mathbf{x}_{C}^{\star}$ with
\begin{equation}
    m=\min_{c\neq d}\|\mathbf{x}_{c}^{\star}
    -\mathbf{x}_{d}^{\star}\|_2>0.
    \label{eq:s5_memory_margin}
\end{equation}
For each cue $c$, suppose the relevant trajectory remains in a region
$\mathcal{B}_{c}$ on which
\begin{equation}
    \|\mathbf{F}(\mathbf{x})-\mathbf{x}_{c}^{\star}\|_2
    \leq\kappa\|\mathbf{x}-\mathbf{x}_{c}^{\star}\|_2,
    \qquad 0\leq\kappa<1.
    \label{eq:s5_memory_contraction}
\end{equation}
Consider
$\mathbf{x}_{t+1}=\mathbf{F}(\mathbf{x}_{t})+\boldsymbol{\eta}_{t}$
and write $\mathbf{e}_{t}=\mathbf{x}_{t}-\mathbf{x}_{c}^{\star}$.
Assume the disturbed trajectory remains in $\mathcal{B}_{c}$ through time
$L$ (almost surely in the stochastic case). Then
\begin{equation}
    \|\mathbf{e}_{L}\|_2
    \leq\kappa^{L}\|\mathbf{e}_{0}\|_2
    +\sum_{t=0}^{L-1}\kappa^{L-1-t}
      \|\boldsymbol{\eta}_{t}\|_2.
    \label{eq:s5_memory_deterministic_bound}
\end{equation}
In particular, if $\|\boldsymbol{\eta}_{t}\|_2\leq\varepsilon$, then
\begin{equation}
    \|\mathbf{e}_{L}\|_2
    \leq\kappa^{L}\|\mathbf{e}_{0}\|_2
    +\varepsilon\frac{1-\kappa^{L}}{1-\kappa}.
    \label{eq:s5_memory_bounded_noise}
\end{equation}
Whenever the right-hand side is smaller than $m/2$, nearest-fixed-point
decoding recovers cue $c$ correctly.

If instead
$\mathbb{E}[\boldsymbol{\eta}_{t}\mid\mathbf{x}_{t}]=\mathbf{0}$ and
$\mathbb{E}[\|\boldsymbol{\eta}_{t}\|_2^2\mid\mathbf{x}_{t}]
\leq\nu^2$, then
\begin{equation}
    \mathbb{E}\|\mathbf{e}_{L}\|_2^2
    \leq\kappa^{2L}\mathbb{E}\|\mathbf{e}_{0}\|_2^2
    +\nu^2\frac{1-\kappa^{2L}}{1-\kappa^2},
    \label{eq:s5_memory_mse_bound}
\end{equation}
and the nearest-fixed-point decoder obeys
\begin{equation}
    \mathbb{P}(\widehat c\neq c)
    \leq\frac{4}{m^2}\left[
    \kappa^{2L}\mathbb{E}\|\mathbf{e}_{0}\|_2^2
    +\nu^2\frac{1-\kappa^{2L}}{1-\kappa^2}
    \right].
    \label{eq:s5_memory_error_bound}
\end{equation}
\end{theorem}

\begin{proof}
Equation~\ref{eq:s5_memory_contraction} and the triangle inequality give
$\|\mathbf{e}_{t+1}\|_2\leq
\kappa\|\mathbf{e}_{t}\|_2+\|\boldsymbol{\eta}_{t}\|_2$.
Induction proves Eq.~\ref{eq:s5_memory_deterministic_bound}, and summing the
geometric series gives Eq.~\ref{eq:s5_memory_bounded_noise}. If
$\|\mathbf{e}_{L}\|_2<m/2$, the triangle inequality implies that
$\mathbf{x}_{L}$ is closer to $\mathbf{x}_{c}^{\star}$ than to every other
fixed point, proving the deterministic decoding claim.

For conditionally zero-mean noise, the cross term between
$\mathbf{F}(\mathbf{x}_{t})-\mathbf{x}_{c}^{\star}$ and
$\boldsymbol{\eta}_{t}$ has conditional expectation zero. Hence
\begin{equation}
    \mathbb{E}\|\mathbf{e}_{t+1}\|_2^2
    \leq\kappa^2\mathbb{E}\|\mathbf{e}_{t}\|_2^2+\nu^2,
\end{equation}
whose iteration proves Eq.~\ref{eq:s5_memory_mse_bound}. A decoding error
requires $\|\mathbf{e}_{L}\|_2\geq m/2$; Markov's inequality applied to
$\|\mathbf{e}_{L}\|_2^2$ proves Eq.~\ref{eq:s5_memory_error_bound}.
\end{proof}

Theorem~\ref{thm:s5_attractor_memory} makes the missing experimental bridge
explicit. Fast local recovery benefits memory only if cue-conditioned states
or trajectories remain separated: contraction controls within-cue dispersion,
whereas $m$ controls between-cue distinguishability. The bias-driven recovery
assay in Fig.~\ref{fig:mechanism} estimates stability around one reference
trajectory but does not measure this separation. For time-dependent
cue-specific reference trajectories satisfying
$\mathbf{x}_{c}^{\star}(t+1)=\mathbf{F}(\mathbf{x}_{c}^{\star}(t))$, the same
proof applies with the time-$L$ margin
$m_L=\min_{c\neq d}\|\mathbf{x}_{c}^{\star}(L)
-\mathbf{x}_{d}^{\star}(L)\|_2$ and the one-step transverse bound
$\|\mathbf{F}(\mathbf{x})-\mathbf{x}_{c}^{\star}(t+1)\|_2
\leq\kappa_t\|\mathbf{x}-\mathbf{x}_{c}^{\star}(t)\|_2$, provided
$\sup_t\kappa_t\leq\kappa<1$.
For unbounded Gaussian disturbances, the stochastic conclusion can be applied
on a high-probability bounded-noise event, with the corresponding basin-exit
probability added to the right-hand side of
Eq.~\ref{eq:s5_memory_error_bound}.

For the nonlinear reservoir in Eq.~\ref{eq:esn}, the Jacobian along a
cue-conditioned reference trajectory is
\begin{equation}
    \mathbf{J}_{t}
    =(1-\alpha)\mathbf{I}
    +\alpha\mathbf{D}_{t}\mathbf{W},
    \qquad
    \mathbf{D}_{t}
    =\operatorname{diag}\!\left[
      \operatorname{sech}^{2}(\mathbf{W}\mathbf{x}_{t}
      +\mathbf{W}_{\mathrm{in}}\widetilde{\mathbf{u}}_{t})
      \right].
    \label{eq:s5_task_jacobian}
\end{equation}
The first-order covariance of terminal deviations generated by isotropic
post-cue process noise is therefore
\begin{equation}
    \boldsymbol{\Sigma}_{L}^{\mathrm{lin}}
    =\sigma_{\xi}^{2}\sum_{k=0}^{L-1}
      \boldsymbol{\Phi}_{L,k+1}\boldsymbol{\Phi}_{L,k+1}^{\top},
    \qquad
    \boldsymbol{\Phi}_{L,k+1}
    =\mathbf{J}_{L-1}\cdots\mathbf{J}_{k+1},
    \label{eq:s5_nonlinear_noise_gramian}
\end{equation}
with the empty product equal to $\mathbf{I}$. Equation
~\ref{eq:s5_nonlinear_noise_gramian} identifies the task-local quantity that
must be measured in the nonlinear network. The commutator defect of
$\mathbf{W}$ alone cannot determine it because both saturation
$\mathbf{D}_{t}$ and cue geometry enter. The post-cue noise scan below
estimates its trace nonlinearly through paired terminal trajectories and
measures cue separation and decoding performance in the same states.

\begin{lemma}[Reciprocity reduces expected antisymmetric energy in the paired-edge generator]
\label{lem:s5_reciprocity_antisymmetry}
Suppose the pre-balance excitatory matrix is loop-free, uses distinct unordered
node pairs, and contains exactly $E$ directed edges, a realized fraction
$\rho$ of which belong to reciprocal pairs. Nonzero weights are
independent and identically distributed with mean $\mu$ and variance $v$,
including the two weights within a reciprocal pair. Ignoring only the
integer-rounding of pair counts,
\begin{equation}
    \mathbb{E}\|\mathbf{W}_{e}-\mathbf{W}_{e}^{\top}\|_{F}^{2}
    =2E\bigl[v+(1-\rho)\mu^{2}\bigr].
    \label{eq:s5_expected_antisymmetry}
\end{equation}
For the $\mathcal{U}(0,1)$ weights used in the generator, this becomes
$E(4-3\rho)/6$ and is strictly decreasing in realized $\rho$. As a function of
the requested continuous parameter before integer rounding, the expectation is
a non-increasing step function.
\end{lemma}

\begin{proof}
There are $E\rho/2$ reciprocal pairs and $E(1-\rho)$ unidirectional edges.
For independent reciprocal weights $X,Y$, the two matrix entries contribute
$2(X-Y)^2$, whose expectation is $4v$. A unidirectional edge of weight $X$
contributes $2X^2$, whose expectation is $2(v+\mu^2)$. Summing the two
contributions gives
$2E\rho v+2E(1-\rho)(v+\mu^2)$, which simplifies to
Eq.~\ref{eq:s5_expected_antisymmetry}. For a uniform weight,
$\mu=1/2$ and $v=1/12$.
\end{proof}

\begin{proposition}[Antisymmetry controls an upper bound on the commutator defect]
\label{prop:s5_commutator_bound}
For any real matrix $\mathbf{W}$, let
$\mathbf{S}=(\mathbf{W}+\mathbf{W}^{\top})/2$ and
$\mathbf{K}=(\mathbf{W}-\mathbf{W}^{\top})/2$. Then
\begin{equation}
    \Delta_{\mathrm{comm}}(\mathbf{W})
    =2\|\mathbf{K}\mathbf{S}-\mathbf{S}\mathbf{K}\|_{F}
    \leq2\|\mathbf{S}\|_{2}
       \|\mathbf{W}-\mathbf{W}^{\top}\|_{F}.
    \label{eq:s5_commutator_bound}
\end{equation}
Consequently,
\begin{equation}
    \delta_{\mathrm{rel}}(\mathbf{W})
    \leq
    \frac{2\|\mathbf{S}\|_{2}
    \|\mathbf{W}-\mathbf{W}^{\top}\|_{F}}{\|\mathbf{W}\|_{F}^{2}}
\end{equation}
for every nonzero $\mathbf{W}$.
\end{proposition}

\begin{proof}
Expanding $\mathbf{W}=\mathbf{S}+\mathbf{K}$ and
$\mathbf{W}^{\top}=\mathbf{S}-\mathbf{K}$ gives
$\mathbf{W}\mathbf{W}^{\top}-\mathbf{W}^{\top}\mathbf{W}
=2(\mathbf{K}\mathbf{S}-\mathbf{S}\mathbf{K})$. The inequality follows from
$\|\mathbf{K}\mathbf{S}\|_{F},\|\mathbf{S}\mathbf{K}\|_{F}
\leq\|\mathbf{S}\|_{2}\|\mathbf{K}\|_{F}$ and
$2\|\mathbf{K}\|_{F}=\|\mathbf{W}-\mathbf{W}^{\top}\|_{F}$.
Dividing by $\|\mathbf{W}\|_F^2$ gives the relative bound.
\end{proof}

Lemma~\ref{lem:s5_reciprocity_antisymmetry} and
Proposition~\ref{prop:s5_commutator_bound} provide a rigorous but deliberately
limited structural bridge. They do not prove that either
$\Delta_{\mathrm{comm}}(\mathbf{W})$ or
$\delta_{\mathrm{rel}}(\mathbf{W})$ decreases for every realization: the symmetric
part also changes, and the independent inhibitory matrix and spectral
normalization are applied after $\mathbf{W}_{e}$ is generated. The monotone
relationship for the final signed matrices is therefore an empirical ensemble
result (Fig.~\ref{fig:mechanism}a), rather than a deterministic consequence of
reciprocity alone.

\begin{proposition}[One-step small-signal reduction of the leaky-tanh network]
\label{prop:s5_nonlinear_reduction}
For
\begin{equation}
    \mathbf{x}_{t+1}=(1-\alpha)\mathbf{x}_{t}
    +\alpha\tanh(\mathbf{W}\mathbf{x}_{t}+\mathbf{B}\mathbf{u}_{t}),
    \qquad 0<\alpha\leq1,
    \label{eq:s5_nonlinear_system}
\end{equation}
set $\mathbf{z}_{t}=\mathbf{W}\mathbf{x}_{t}+\mathbf{B}\mathbf{u}_{t}$ and
$\mathbf{A}=(1-\alpha)\mathbf{I}+\alpha\mathbf{W}$. If
$\|\mathbf{z}_{t}\|_{\infty}\leq\delta$, then
\begin{equation}
    \mathbf{x}_{t+1}=\mathbf{A}\mathbf{x}_{t}
    +\alpha\mathbf{B}\mathbf{u}_{t}+\mathbf{e}_{t},
    \qquad
    \|\mathbf{e}_{t}\|_{2}\leq\frac{\alpha\sqrt{N}}{3}\delta^{3}.
    \label{eq:s5_tanh_remainder}
\end{equation}
The trajectory Jacobian is
\begin{equation}
    \mathbf{J}_{t}=(1-\alpha)\mathbf{I}
    +\alpha\mathbf{D}_{t}\mathbf{W},
    \qquad
    \mathbf{D}_{t}=\operatorname{diag}(\operatorname{sech}^{2}z_i(t)),
\end{equation}
and obeys
\begin{equation}
    \|\mathbf{J}_{t}-\mathbf{A}\|_{2}
    \leq\alpha\tanh^{2}(\delta)\|\mathbf{W}\|_{2}.
    \label{eq:s5_jacobian_bound}
\end{equation}
In addition, the cube $[-1,1]^N$ is forward invariant.
\end{proposition}

\begin{proof}
For scalar $z$,
$|\tanh z-z|=|\int_{0}^{z}\tanh^{2}s\,ds|
\leq\int_{0}^{|z|}s^{2}\,ds=|z|^{3}/3$ because
$|\tanh s|\leq|s|$. Applying this coordinatewise proves
Eq.~\ref{eq:s5_tanh_remainder}. Since
$1-\operatorname{sech}^{2}z=\tanh^{2}z$,
$\|\mathbf{D}_{t}-\mathbf{I}\|_{2}\leq\tanh^{2}\delta$, which proves
Eq.~\ref{eq:s5_jacobian_bound}. Finally, if
$\|\mathbf{x}_{t}\|_{\infty}\leq1$, then
$\|\mathbf{x}_{t+1}\|_{\infty}
\leq(1-\alpha)\|\mathbf{x}_{t}\|_{\infty}+\alpha\leq1$.
\end{proof}

\begin{remark}[Scope at the normalization boundary]
\label{rem:s5_boundary}
The infinite-horizon statements above require
$\varrho(\mathbf{A})<1$. The main simulations instead normalize
$\mathbf{W}$ so that its spectral abscissa
$\max_i\operatorname{Re}\lambda_i(\mathbf{W})$ is $1$. Since every
eigenvalue $\lambda_i(\mathbf{W})$ maps to
$1-\alpha+\alpha\lambda_i(\mathbf{W})$, an eigenvalue with real part $1$
maps to $1+i\alpha\operatorname{Im}\lambda_i$, whose modulus is at least
$1$. Thus the zero-state linear map used for comparison is on or beyond the
stability boundary, and its stationary Gramian need not exist. Finite-horizon
Gramians remain well defined, while Proposition~\ref{prop:s5_nonlinear_reduction}
quantifies the linear approximation only when pre-activations are small. The
global tanh nonlinearity guarantees bounded states through forward invariance,
but boundedness alone does not prove a unique stationary distribution or a
global SSV ordering. Consequently, the theorems establish a rigorous local and
constructive mechanism; the monotone relationship between $\rho$ and SSV in the full
nonlinear network remains an experimentally tested ensemble result.
\end{remark}

\subsection*{S6. E/I balance and scale-invariant non-normality controls}
\label{sec:supp_ei_controls}

The original implementation balances the summed excitatory and inhibitory
weights globally. Under the stated parameters,
\texttt{calculate\_inhibitory\_connection\_probability(300)} returns one, so
the independently sampled inhibitory support is dense off the diagonal. To
test whether the relationship between $\rho$ and SSV depends on this choice, we compared
the original global-cancellation network with a row-balanced signed network
and a Dale-like network in which presynaptic neurons retain fixed signs and
inhibitory inputs enforce row-wise balance. All controls used max-modulus
normalization, $N=300$, $c=0.1$, $T=4{,}000$, a 500-step washout, and $n=10$
networks per $\rho$ level.

SSV decreased with $\rho$ under all three schemes (correlations of the 11
level means: global $r=-0.975$, row-balanced $r=-0.970$, and Dale-like
$r=-0.977$; all $p<7.5\times10^{-7}$; Supplementary
Fig.~\ref{fig:supp_ei}). Thus the principal relationship between $\rho$ and SSV is not an
artefact of dense global inhibition or the absence of row-wise balance. The
relative commutator defect fell by endpoint ratios of $1.419$ and $1.417$ in
the global and row-balanced schemes, respectively. In the Dale-like control,
however, it changed by only $1.015\times$ despite the large SSV decline
($r=-0.918$, $p=6.8\times10^{-5}$ for the small but consistent monotone trend
in level means).
Moreover, after centering within each $\rho$ level, the
correlation between $\delta_{\mathrm{rel}}$ and SSV was not significant in any scheme
($|r|\leq0.145$, all $p\geq0.130$). These controls support the robustness of
the structural SSV effect while showing that the relative commutator defect is
an ensemble-specific correlate, not a universal or sufficient mediator across
E/I architectures.

\begin{figure}[htbp]
		\centering
		\includegraphics[width=\textwidth]{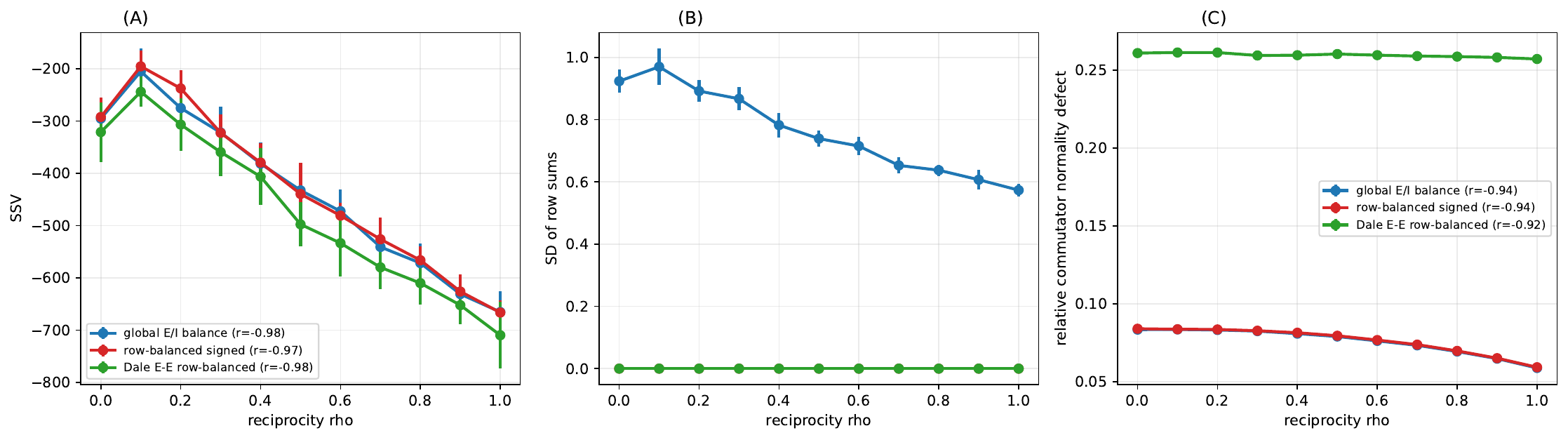}
	\caption{\textbf{E/I-balance controls.}
		\textbf{(a)} SSV versus reciprocity under global cancellation,
		row-balanced signed weights, and a Dale-like row-balanced construction;
		legend values are correlations of the 11 level means.
		\textbf{(b)} Standard deviation of row sums. Values near zero verify the
		two explicit row-balance controls.
		\textbf{(c)} Relative commutator normality defect
		$\delta_{\mathrm{rel}}$ versus reciprocity; legend values are the
		corresponding level-mean correlations. Points and error bars show
		mean $\pm$ SD over $n=10$ networks per level.}
		\label{fig:supp_ei}
\end{figure}

\subsection*{S7. Shared-matrix robustness and operating-point controls}
\label{sec:supp_order_parameter}

The confirmatory experiment was designed to prevent task-specific network
selection and post hoc readout adaptation. \newrev{Online memory capacity}, autonomous restoration,
persistent-pattern maintenance, and evidence integration reused the identical
recurrent-matrix checksum for each network ID. Input, cue, trial, and
perturbation seeds were generated from a resolved master seed and saved with
the raw rows. Zero-perturbation trajectory identity, requested endpoint
reciprocity, cross-task matrix hashes, and rerun determinism were checked
automatically before analysis.

The full evidence-noise and state-perturbation curves are shown in
Supplementary Fig.~\ref{fig:supp_integration_robustness}. Under the primary
maximum-real normalization, input-noise robustness AUC had no network-level
trend with reciprocity ($r_s=-0.093$, $p=0.520$). At the independently frozen
$0.025\times$ state-perturbation dose, above-chance performance retention was
also non-significant ($r_s=0.084$, $p=0.563$), despite a strong decrease in the
distance AUC between paired perturbed and reference task trajectories
($r_s=-0.903$, $p=3.4\times10^{-19}$). The distinction is substantive: two
trajectories can converge geometrically while crossing a decision boundary or
losing between-class margin.

The auxiliary repeated-pattern assay likewise did not provide an independent
multi-attractor result. At the longest delay in the primary condition,
network-level trends in cue-separation ratio ($r_s=0.228$, $p=0.111$), bit
accuracy ($r_s=0.158$, $p=0.274$), and NMSE ($r_s=-0.116$, $p=0.424$) were not
significant. Under maximum-modulus normalization, pattern measures did improve
with reciprocity (bit accuracy $r_s=0.721$, $p=7.1\times10^{-6}$; NMSE
$r_s=-0.750$, $p=1.8\times10^{-6}$), whereas autonomous restoration did not
($r_s=0.047$, $p=0.805$). Together with the primary result, this crossover
shows that reciprocity interacts with spectral normalization and task geometry;
no single stability metric is sufficient to predict every delayed computation.

\begin{figure}[htbp]
	\centering
	\includegraphics[width=\textwidth]{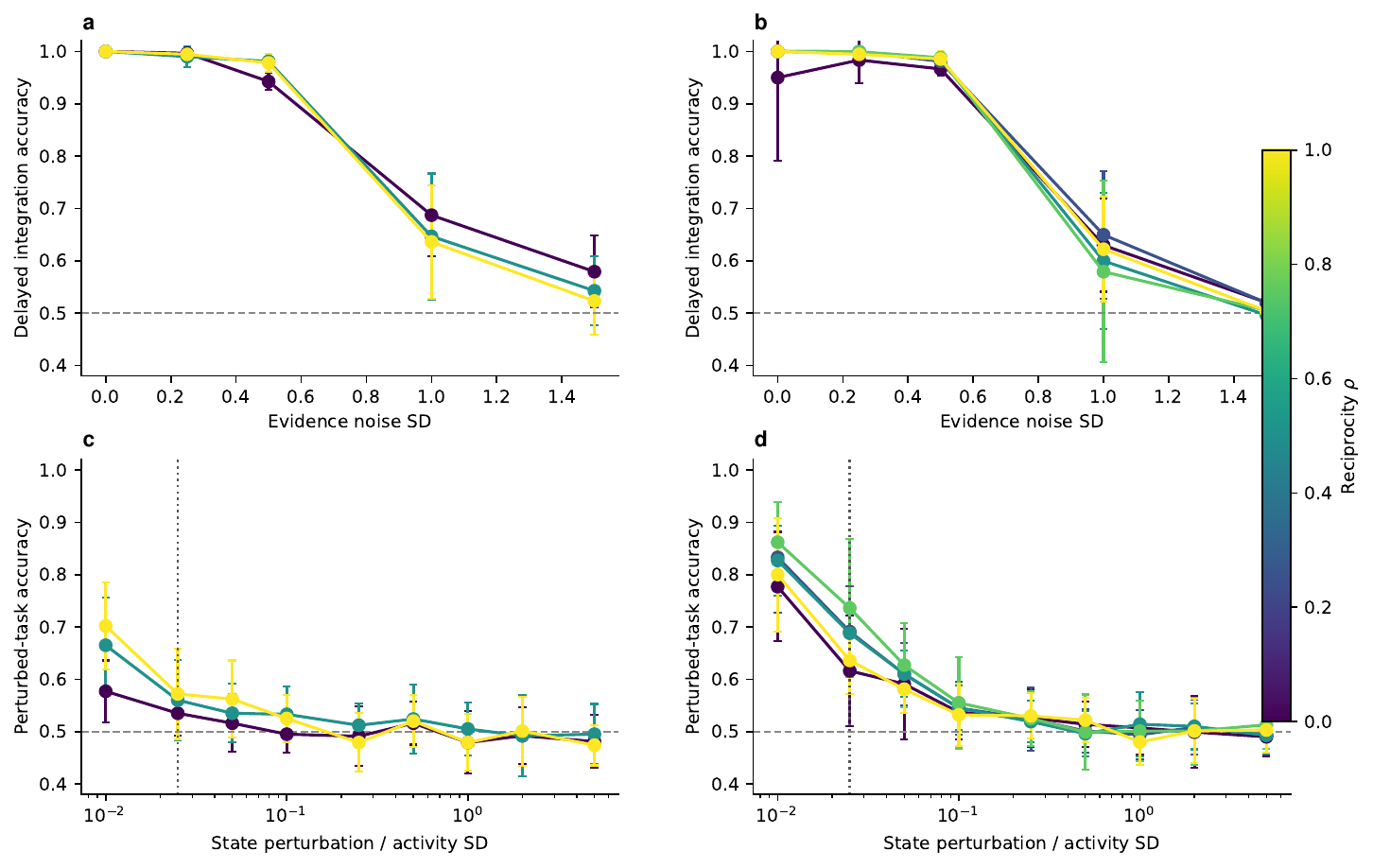}
	\caption{\textbf{Noisy evidence-integration and perturbation-dose controls.}
		\textbf{(a,b)} Frozen-readout accuracy across evidence-noise SD for the
		paired maximum-modulus control (a) and primary maximum-real condition (b).
		\textbf{(c,d)} Accuracy after a midpoint state perturbation, scaled by the
		network's activity SD, under the corresponding control (c) and primary (d)
		conditions. Colours denote reciprocity; points and error
		bars are mean $\pm$ SD over independent networks. The vertical dotted line
		marks the $\rho$-blinded pilot dose ($0.025\times$ activity SD) frozen before
		the confirmatory run; horizontal dashed lines mark chance accuracy. The
		curves do not support a universal monotonic robustness benefit of larger
		$\rho$.}
	\label{fig:supp_integration_robustness}
\end{figure}

\subsection*{S8. Post-cue process-noise attenuation and the cue-margin boundary}
\label{sec:supp_dms_process_noise}

The post-cue DMS scan was designed after the fixed-spectrum theory had
identified two distinct requirements for reliable delayed decoding: low
noise amplification and preserved cue separation. It reused the audited
seeded-replication matrices but used a new master seed for cues, input weights,
and process-noise directions. All 210 expected network files were produced;
source matrix hashes, resolved configuration hashes, dose vectors, paired
seeds, predictions, and finite-valued outputs passed automated checks.
\rev{For the noise-matched analysis, the exactly zero-noise fit was treated as
a diagnostic because its within-class covariance is degenerate; retention was
therefore normalized to the smallest positive dose and its AUC was computed
over positive doses only.}

The state-space prediction was unambiguous. At $\sigma_x=10^{-4}$, empirical
gain decreased from $12.25\pm3.88$ at $\rho=0$ to $4.79\pm0.35$ at
$\rho=1$ (network-level $r_s=-0.863$, $p=1.7\times10^{-63}$; correlation of
the 21 level means $r=-0.898$, $p=3.3\times10^{-8}$). The paired endpoint
difference was $-7.46$ (bootstrap 95\% CI $[-9.81,-5.16]$; paired
$t$-test $p=2.1\times10^{-4}$). This directly confirms that, in the original
maximum-real operating regime, larger $\rho$ attenuates deviations generated by
isotropic post-cue process noise.

The behavioural prediction remained conditional. Frozen-readout retention
AUC increased weakly over the full network ensemble ($r_s=0.226$,
$p=9.7\times10^{-4}$), with seven of ten seed-block slopes positive
(one-sample $t$-test $p=0.026$), but the rho endpoints did not differ:
$0.534\pm0.210$ at $\rho=0$ and $0.502\pm0.106$ at $\rho=1$ (difference
$-0.032$, 95\% CI $[-0.200,0.114]$). Allowing shrinkage LDA to match each
noise covariance also produced no monotonic network-level trend
\rev{($r_s=0.077$, $p=0.264$).} Thus the absence of a fixed-readout endpoint effect
cannot be attributed solely to readout mismatch.

Cue geometry explains the boundary. Clean minimum centroid margin declined
across $\rho$-level means ($r=-0.575$, $p=0.0064$), although the noisier
network-level rank test did not reach significance ($r_s=-0.122$, $p=0.078$).
For a common covariance, Eq.~\ref{eq:s5_multiclass_error_bound} depends on a
Mahalanobis rather than Euclidean margin. Because
$\boldsymbol{v}^{\top}\boldsymbol{\Sigma}^{-1}\boldsymbol{v}
\geq\|\boldsymbol{v}\|_2^2/\operatorname{tr}\boldsymbol{\Sigma}$, the
measured quantity $m_{\mathrm{clean}}^2/G_L^{\mathrm{emp}}$ is a scaled
lower-bound proxy at fixed noise dose, not an exact Bayes margin. It predicted
frozen-readout AUC strongly ($r_s=0.722$, $p=4.2\times10^{-35}$) but did not
itself vary monotonically with $\rho$ ($r_s=-0.020$, $p=0.773$). The scan
therefore validates high-$\rho$ intrinsic noise attenuation while showing why
that property alone is insufficient for universal behavioural robustness.

\end{document}